\documentclass[10pt,journal,compsoc]{IEEEtran}

\usepackage{graphicx}
\usepackage{algorithmic}
\usepackage{subfigure}
\usepackage{algorithm}
\usepackage{xcolor}
\usepackage{mathcomp} 
\usepackage{eqparbox} 
\usepackage{amsfonts}	
\usepackage{multirow}
\usepackage{amssymb}
\usepackage{amsmath}

\newtheorem{definition}{Definition}[section]
\newtheorem{theorem}{Theorem}[section]

\newtheorem{proof}{Proof}[section]
\newtheorem{lemma}{Lemma}[section]

\begin{document}

\title{Reachability Queries with Label and Substructure Constraints on Knowledge Graphs}


\author{Xiaolong~Wan, Hongzhi~Wang
\IEEEcompsocitemizethanks{\IEEEcompsocthanksitem The authors are with the School of Computer Science and Technology, Harbin Institute of Technology, China.\protect\\
E-mail: wxl@hit.edu.cn, wangzh@hit.edu.cn
}
\thanks{This work has been submitted to the IEEE for possible publication. Copyright may be transferred without notice, after which this version may no longer be accessible.}}

\markboth{IEEE Transactions on Knowledge and Data Engineering,~Vol.~XX, No.~XX, XX~20XX}%
{Wan \MakeLowercase{\textit{et al.}}: Reachability Queries with Label and Substructure Constraints on Knowledge Graphs}


\IEEEtitleabstractindextext{%
\begin{abstract}
	Since knowledge graphs~(KGs) describe and model the relationships between entities and concepts in the real world, reasoning on KGs often correspond to the \underline{r}eachability queries with \underline{l}abel and \underline{s}ubstructure \underline{c}onstraints~(LSCR). Specially, for a search path $p$, LSCR queries not only require that the labels of the edges passed by $p$ are in a certain label set, but also claim that a vertex in $p$ could satisfy a certain substructure constraint. LSCR queries is much more complex than the label-constraint reachability~(LCR) queries, and there is no efficient solution for LSCR queries on KGs, to the best of our knowledge. Motivated by this, we introduce two solutions for such queries on KGs, UIS and INS. The former can also be utilized for general edge-labeled graphs, and is relatively handy for practical implementation. The latter is an efficient \textit{local-index}-based informed search strategy. An extensive experimental evaluation, on both synthetic and real KGs, illustrates that our solutions can efficiently process LSCR queries on KGs.
\end{abstract}

}

\maketitle

\IEEEdisplaynontitleabstractindextext

\IEEEpeerreviewmaketitle



\IEEEraisesectionheading{\section{Introduction}}\label{section:introduction}
	\IEEEPARstart{N}{owadays}, numerous highly formatted databases are utilized to construct domain specific KGs \cite{SenDHCHRLR18Fonduer}, reasoning on which is widely used in various applications, e.g. criminal link analysis~\cite{Schroeder2003CrimeLink, SchroederXCC07}, suspicious transaction detection~\cite{RajputKLH14}, e-commerce recommendation~\cite{RecommendationAlgorithm}, etc. In such applications, reasoning between two entities often relates to the reachability queries with both label and substructure constraints~\cite{SchroederXCC07, XuC04}, as KGs can be considered as edge-labeled graphs. 
	
	For example, we model a financial KG as $G$, where (\romannumeral1) each vertex represents a person; (\romannumeral2) each edge $e$$=$$(u,l,v)$ correlates to an account transfer or a relationship from vertex $u$ to vertex $v$; (\romannumeral3) the label $l$ of $e$ is either a timestamp that corresponds to the transfer occurred time, or a social relationship~(e.g. \textit{friend-of}, \textit{parent-of}, \textit{married-to}, etc.). In some detection tasks, to verify the economic criminal relationship between Suspect $\mathcal{C}$ and Suspect $\mathcal{P}$, with the known information: \textit{``An indirect transaction from $\mathcal{C}$ to $\mathcal{P}$ occurred in April 2019, in which one of the middlemen of the transaction and Amy are married ...''}, as depicted in Figure~\ref{fig:person_figure}(a).

	If the economic criminal relationship exists, a transaction path $p=<\mathcal{C},e_0,v_0,\dots,v_{k-1},e_k,\mathcal{P}>$ from $\mathcal{C}$ to $\mathcal{P}$ exists in $G$ as shown in Figure~\ref{fig:person_figure}(b), where (\romannumeral1) $l_{0}$, ..., $l_{k}$ represent the labels of edges $e_0$, ..., $e_{k}$, respectively, and the timestamps of $l_{0}$, ..., $l_{k}$ are in April 2019; (\romannumeral2) $\exists \mathcal{X}\in v_0,\dots,v_{k-1}$, which is one of the passing vertex of $p$, represents a person in $G$ who married $Amy$. Note that $G$ could contain multiple married couples, maybe only some of those couples could satisfy the above conditions.
	
	
	\begin{figure}
		\centering
		\renewcommand{\thesubfigure}{}
		\subfigure[(a)]{
			\includegraphics[scale = 0.42]{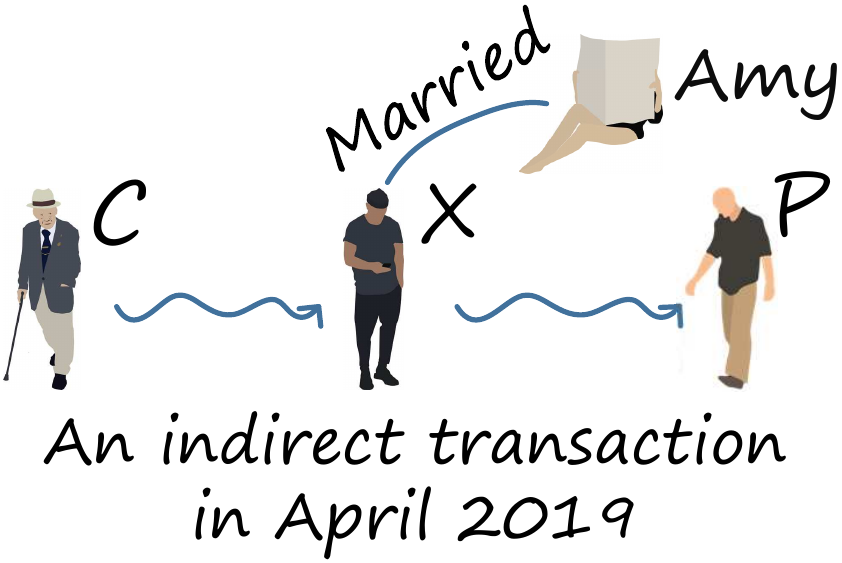}}
		\subfigure[(b)]{
			\includegraphics[scale = 0.47]{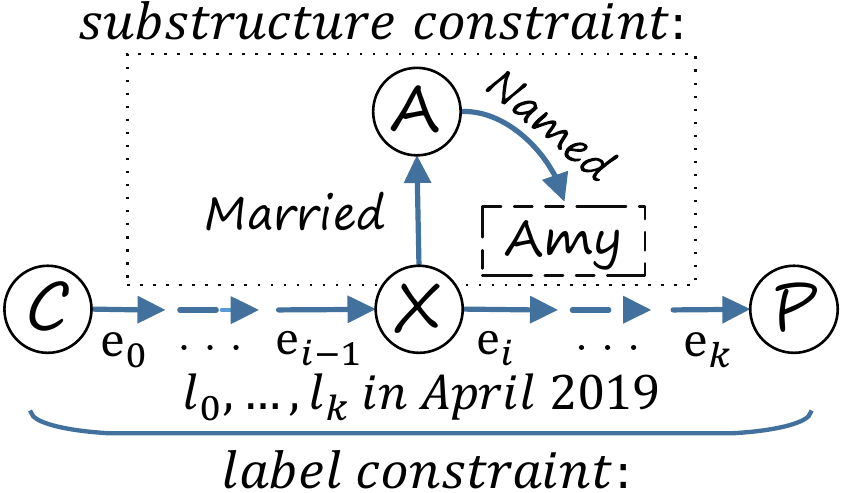}}
		\caption{A financial crime detection scenario query example.}
		\label{fig:person_figure}
	\end{figure}
	
	Moreover, such queries can also be used in a variety of situations, e.g. traffic navigation~\cite{FTSPWUR}, domain specific anomaly detection~\cite{AnomalyDetectionSurvey}, etc., asking the following question: \textit{Can a vertex $s$ reach to a vertex $t$ by a path $p$, where the edge labels in $p$ are in a certain label set~(label constraint), and $p$ passes a certain substructure~(substructure constraint)?}
	
	Even though the \underline{l}abel- and \underline{s}ubstructure-\underline{c}onstraint \underline{r}eachability (LSCR) queries on KGs are quite useful for the above reasoning tasks, as far as we know, there is no particular solution.  The reason is that the LSCR problem is much more complicated than the label-constraint reachability~(LCR) problem, originated in \cite{Jin2010CLR}. According to \cite{Jin2010CLR}, there are two kinds of techniques for LCR queries, \textit{online search (DFS/BFS)} and \textit{full transitive closure~(TC) precomputing}.
	
	The former can be applied directly, for the LCR queries. Assuming the input KG contains $|V|$ vertices and $|E|$ edges, the time complexity of DFS/BFS applied in LCR queries is $O(|V|+|E|)$, as the label constraint prunes the search space of DFS/BFS. Unfortunately, the problem of LSCR queries on KGs is too complex. As analyzed in Section~\ref{section:uninformed_search}, the time complexity of implementing traditional BFS/DFS for LSCR queris is about $O\big(|V|\times(|V|+|E|)\big)$ for the worst case. Thus, the former techniques can hardly be adopted to solve our problem.
		
	The latter aims to precompute a full TC of the input graph to answer LCR queries efficiently, but, the space complexity of storing a full TC for a graph $G$ is $O(|V|^2\times 2^{\mathcal{L}})$, where $\mathcal{L}$ is the label set of $G$. Hence, many works~\cite{Jin2010CLR, ZouXY0XZ14, Valstar17Landmark} are proposed to reduce the space complexity of indexing an input graph. However, the indexing time complexity of the above methods are too high to index a large KG, as discussed in Section~\ref{section:onlineSearch:LCR}. Note that, in this paper, we do not consider the approximate techniques of LCR queries, e.g. \cite{Nehaicde19ARROW}, because KG reasoning applications require exact results, like the above criminal detection scenario.
	
	\textbf{Contributions.} In this paper, we provide a comprehensive study of LSCR queries on KGs with efficient solutions. Our work conquers the practical unbearable indexing time of the traditional landmark indexing method~\cite{Valstar17Landmark}, with a \textit{local index}. Plus, we overcome the inefficiency of the algorithm~(Section~\ref{section:improved_u_s}) that implements the subgraph matching techniques for LSCR queries on KGs, by developing an informed search strategy based on local index. With the detailed experimental evaluation, our work could address the LSCR queries on KGs efficiently.
	
	Our main contributions are as follows:
	\begin{itemize}
		\item[-] We first define the problem of reachability queries with label and substructure constraints, which is the basis of reasoning on KGs.
		
		
		\item[-] We develop an \underline{u}n\underline{i}nformed \underline{s}earch algorithm, named UIS, for LSCR queries. It has a great potential for practical implementation on general labeled graphs, as UIS is not overly complex, and does not require any of the characteristics of KGs.
		
		
		
		\item[-] Then, we develop an intuition idea of implementing SPARQL engines~\cite{SurveyOfSPARQLEngines}, called UIS$^*$, to address LSCR queries and overcome the disadvantage of UIS. That is because: (\romannumeral1) current KGs are based on the standard of OWL~\cite{OWL} in the form of RDF~\cite{RDF-3X}, usually can be accessed by SPARQL; (\romannumeral2) a substructure constraint can be expressed in SPARQL~(Section~\ref{section:background}).

		
		\item[-] After that, we propose an efficient \underline{in}formed \underline{s}earch algorithm, i.e. INS, with a novel index, \textit{local index}. INS could break the fixed search direction of the uninformed search algorithms, and could process the vertices in the input KG with priorities, by developing two data structures as an evaluation function based on the entries of \textit{local index}.  
		
		
		\item[-] Finally, we conduct extensive experiments on both real and synthetic datasets to evaluate the introduced search algorithms. The experimental results show that the informed search algorithm could address LSCR queries on KGs efficiently.
		
	\end{itemize}
	
	The rest of this paper is organized as follows. The problem of LSCR queries on KGs is formally defined in Section~\ref{section:background}. We introduce UIS, UIS$^*$ and INS in Section~\ref{section:uninformed_search}, Section~\ref{section:improved_u_s} and Section~\ref{section:informed_search}, respectively. The experimental results are presented in Section~\ref{section:result}. Finally, we draw a conclusion in Section~\ref{section:conclusion}.

\section{Problem Definition}\label{section:background}

	We formally define the LSCR queries on KGs, with a table of frequently used notations~(Table~\ref{tab:notation}) and two KG schematic views~(Figure~\ref{fig:RDFS} and Figure~\ref{fig:subgraph_example}).
	
	\textbf{Knowledge Graph~(KG).} In essence, KGs~\cite{Freebase, YAGO} are graphs that consist of vertices and labeled edges. However, different from normal edge-labeled graphs, KGs often correspond to rich information of the real world, and the sizes of KGs are relatively large. For example, at the end of 2014, Freebase~\cite{Freebase} included 68 million entities, 1 billion pieces of relationship information and more than 2.4 billion factual triples. In practice, KGs are stored by RDF triples~\cite{RDF-3X}, and formatted by RDFS~\cite{RDFS-Reasoning-Parallel}~(RDF schema). RDFS defines a set of RDF vocabularies~(e.g. ``rdfs:Class'', ``rdfs:subClassOf'', ``rdf:type'') with special meanings, which could structure RDF resources and format the knowledge representations. From the perspective of graph classification, RDFS represents KGs as scale-free networks~\cite{LKAQ}, in which the relative commonness of vertices with a degree greatly exceeds the average. 
	
	Figure~\ref{fig:RDFS} draws a simple KG with some RDF vocabularies, in which (\romannumeral1) ``eg:Researcher'' represents a class of the researchers in the real world; (\romannumeral2) ``eg:workWith'' is an edge label entity, and the edges labeled by ``eg:workWith'' are incident to~(``rdfs:domain'') and also points to~(``rdfs:range'') an instance of ``eg:Researcher''; (\romannumeral3) Taylor and Walker are real persons and all researchers~(instances of ``eg:Researcher''), who also work together~(with the edges labeled by ``eg:workWith''). Certainly, in KGs, a vertex $u$ corresponding to real word instances has high degree, like ``eg:Researcher''. Besides, a class vertices, like ``eg:Person'' is more important for the KG structure than an instance vertex, e.g. ``eg:Researcher''.

	In a word, KGs are basically edge-labeled graphs, in the form of RDF, and structured by RDFS. The formal definition of a KG is demonstrated in Definition~\ref{definition:kg}. 
	
	\begin{figure}[t]
		\centerline{\includegraphics[scale=0.6]{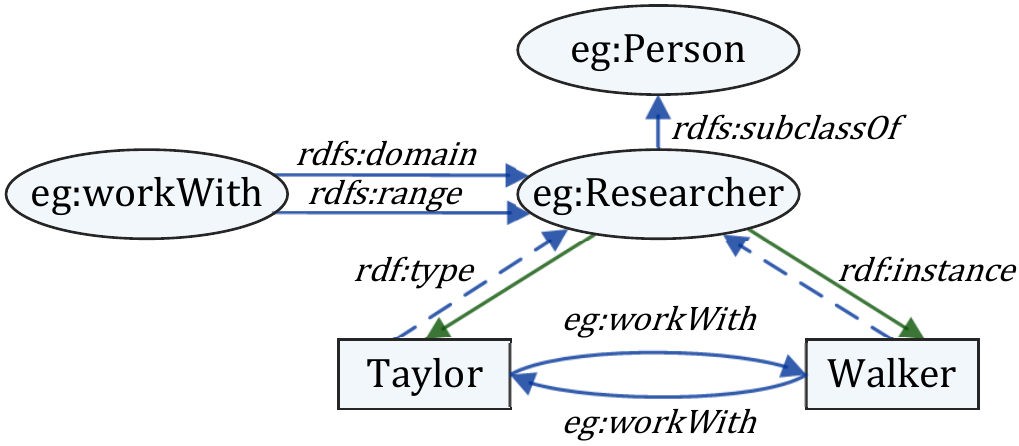}}
		\caption{An example knowledge graph, which is formatted by RDFS.} 
		\label{fig:RDFS}
	\end{figure} 
	
	\begin{definition}
		\label{definition:kg}
		A \textit{knowledge graph} $G$ is a quadruple $(V, E, \mathcal{L}$, $L_{S})$, where: (1) $V$ is the vertex set of $G$ and $\mathcal{L}$ is the edge label set; (2) $E\subseteq V\times\mathcal{L}\times V$ is the labeled edge set of $G$, and for $\forall e\in E$, $e=(s,l,t)$ is an edge from $s$ to $t$ with a label $l$; (3) Let $\lambda$: $E\rightarrow \mathcal{L}$ be a label mapping function, i.e. $\lambda(e)=l$, then $\mathcal{L}=\bigcup_{\forall e\in E}(\lambda(e))$; (4) $L_{S}$ constains the RDFS triples of $G$.
	\end{definition}
	Note that, we illustrate the following definitions based on an input KG $G\negmedspace=\negmedspace(V, E, \mathcal{L}, L_{S})$.
	
	\textbf{Label constraint.} A label constraint in this paper relates to a subset of $\mathcal{L}$ in the KG $G$, for example, in Figure~\ref{fig:subgraph_example}(a), a label constraint can be $\{friendOf,follows\}$.
	
	\textbf{Substructure Constraint.} In order to define the substructure constraints in the KG $G$, we first formalize a variable-substructure $\gamma$ in $G$ as a triple $(V_\gamma,E_\gamma, E_?)$, where (\romannumeral1) $V_\gamma\negmedspace\subseteq\negmedspace V$, $E_\gamma\negmedspace\subseteq\negmedspace E$ and $\forall(u,l,v)\negmedspace\in\negmedspace E_\gamma$, $u,v\in V_\gamma$; (\romannumeral2) $E_?\negthickspace=\negthickspace\big\{\{(u,l,?v)|u\in V_\gamma\wedge l\in\mathcal{L}\}\cup\{(?u,l,v)|v\in V_\gamma\wedge l\in\mathcal{L}\}\big\}$, and an element $e$ in $E_?$ matches zero or multiple edges in $E$. Then, the substructure constraint is defined as follows:
	
	\begin{table}[t]
		\centering \caption{Frequently-used notations.} \label{tab:notation}
		\small
		\begin{tabular}{cp{6.7cm}}
			\hline
			\textbf{Notation} & \textbf{Description}\\%
			\hline
			$\mathbb{V}(p)$ & The vertex set of a path $p$.\\
			$\mathbb{L}(p)$ & The label set of a path $p$.\\
			$P(s,t)$ & All the paths from a vertex $s$ to a vertex $t$.\\
			$M(s,t)$ & The CMS from a vertex $s$ to a vertex $t$.\\
			$V(S,G)$ & A vertex set, containing all the vertices in a KG $G$ that all satisfy a substructure constraint $S$.\\
			\hline
		\end{tabular}
	\end{table}
	
	\begin{definition}
		A substructure constraint $S$ is referred to a tuple $(?x, V_S,E_S, E_?)$, where $(V_S,E_S, E_?)$ corresponds to a variable-substructure in $G$, and, $\exists e\in E_?$, $e$ is incident to $?x$ or points $?x$.
	\end{definition}
	For instance, Figure~\ref{fig:subgraph_example}(b) demonstrates a substructure constraint $S_0$ of $G_0$ in Figure~\ref{fig:subgraph_example}(a), where  $S_0=(?x,\{v_3\},\{\},\\\{(?x,friendOf,v_3),(v_3,likes,?y)\})$.
	
		Note that, although the substructure constraints can be defined in multiple ways, we only study the most widely used one, in this paper, and the other forms can be derived from this definition.
		
	\begin{figure}[t]
		\centerline{\includegraphics[scale=0.4]{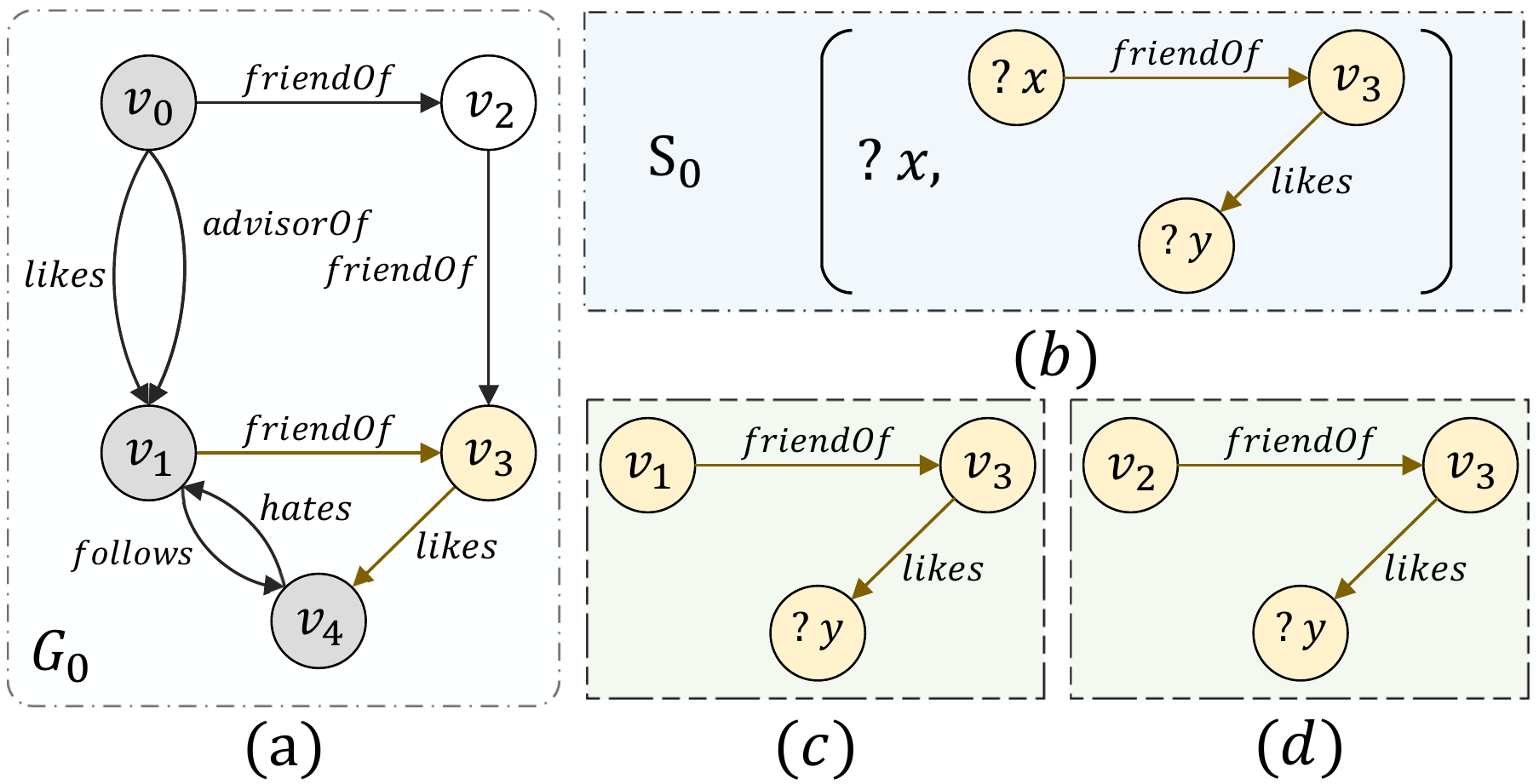}}
		\caption{A schematic view of a running example.} 
		\label{fig:subgraph_example}
	\end{figure} 		
	
	\textbf{Reachability.} For clarity, we let $\mathbb{L}(p)$ and $\mathbb{V}(p)$ denote the label set and the vertex set of a path $p$, respectively, and let $P(s,t)$ represent the set of all the paths from a vertex $s$ to a vertex $t$. Obviously, if $s$ reaches $t$, there is a path from $s$ to $t$~($P(s,t)\negthickspace\neq\negthickspace\phi$). We utilize symbol $s\negmedspace\leadsto\negmedspace t$ and symbol $s\negmedspace \nrightarrow\negmedspace t$ to describe the reachability and non-reachability from a vertex $s$ to a vertex $t$, respectively.

	\textbf{Label-constraint reachability~(LCR).} Similarly, if a vertex $s$ reaches a vertex $t$ under a label constraint $L$ (denoted by $s\negmedspace\stackrel{L}{\leadsto}\negmedspace t$), one path $p$ form $s$ to $t$ exists, where each edge label in $p$ is an element of $L$, i.e. $\mathbb{L}(p)\subseteq L$. 
	
	In contrast, given a label constraint $L$, if $s\negmedspace\stackrel{L}{\leadsto}\negmedspace t$, we say $L$ is a \textit{sufficient path label set}. Furthermore, if $\forall L'\negmedspace\subset\negmedspace L\wedge s\negmedspace\stackrel{L'}{\nrightarrow}\negmedspace t$, $L$ is a \textit{minimal sufficient path label set}. For the \textit{\underline{c}ollection of all the \underline{m}inimal \underline{s}ufficient path label set}s~(\textbf{CMS}) from $s$ to $t$ is  denoted by $M(s,t)$, and defined in Definition~\ref{definition:minimal_path_label_set}.

	\begin{definition}
		\label{definition:minimal_path_label_set}
		A collection of label sets  $M(s,t)$ is a CMS from $s$ to $t$, iff $M(s,t)=\{\mathbb{L}(p_i)|p_i\in P(s,t)\wedge\nexists p_j\in P(s,t)$, $i\neq j$, $such$ $that$ $\mathbb{L}(p_j)\subset \mathbb{L}(p_i)\}$.
	\end{definition}
	For example, $M(v_0,v_3)$$=$$\{\{friendOf\}\}$ and $M(v_0,v_4)=\{\{friendOf,likes\},\{advisorOf,follows\},\{likes, follows\}\}$, as shown in  Figure~\ref{fig:subgraph_example}(a).
	
	\textbf{Substructure-constraint reachability.} Firstly, we say a vertex $u$ \textit{satisfies} a substructure constraint $S\negthickspace=\negthickspace(?x, V_S, E_S$, $E_?)$, if we replace $?x$ with $u$ whose result is still a substructure or a variable-substructure of $G$. For instance, for the substructure constraint $S_0$~(Figure~\ref{fig:subgraph_example}(b)) and two vertices $v_1$ and $v_2$, the replaced results are depicted in Figure~\ref{fig:subgraph_example}(c) and Figure~\ref{fig:subgraph_example}(d), respectively. As the replaced results of $v_1$ and $v_2$ are a variable-substructure of $G_0$, $v_1$ and $v_2$ satisfy $S_0$.  
	
	Then, if a vertex $s$ reaches a vertex $t$ under a substructure constraint $S$~(denoted by $s\negmedspace\stackrel{S}{\leadsto}\negmedspace t$), there is a path $p$ from $s$ to $t$, where $\exists u\negmedspace\in\negmedspace \mathbb{V}(p)$ and $u$ satisfies $S$. For example, in Figure~\ref{fig:subgraph_example}, $v_0\stackrel{S_0}{\leadsto}v_4$,  $v_0\stackrel{S_0}{\leadsto}v_3$ and $v_3\stackrel{S_0}{\leadsto}v_4$.
	
	\textbf{Overall.} Theorem~\ref{theorem:LSCR} presents the condition for the existence of $s \negmedspace\stackrel{L,S}{\leadsto}\negmedspace t$, which states that a vertex $s$ reaches a vertex $t$ under a label constraint $L$ and a substructure constraint $S$. For example, in Figure~\ref{fig:subgraph_example}(a), given a label constraint $L=\{likes,follows\}$, $v_0\negmedspace\stackrel{L,S_0}{\leadsto}\negmedspace v_4$, while, $v_0\negmedspace\stackrel{L,S_0}{\nrightarrow}\negmedspace v_3$. 
	
	\begin{theorem}
		\label{theorem:LSCR}
		$s\negmedspace\stackrel{L,S}{\leadsto}\negmedspace t$ in a KG $G$, iff, $\exists p\in P(s,t)$, $\mathbb{L}(p)\subseteq L$ and $\exists u\in \mathbb{V}(p)\wedge u\ satisfies\ S$.
	\end{theorem}
	\begin{proof}
		That can be easily proved by the above definitions.
	\end{proof}

	\textbf{Problem Definition.} According to the above definitions, we define the problem in Definition~\ref{definition:LSCR}.
	
	\begin{definition}
		\label{definition:LSCR}
		A LSCR query $Q$, on KG $G$$=$$(V,E,\mathcal{L}, L_S)$, is a quadruple $(s,t,L,S)$, where $s,t\in V$, $L\subseteq\mathcal{L}$, and $S$ is a substructure constraint of $G$. If $s\negmedspace\stackrel{L,S}{\leadsto}\negmedspace t$ in $G$, then the answer of $Q$ is true. Otherwise, the answer of $Q$ is false. 
	\end{definition}
	
	\textbf{Additional.} For the KG $G$ and a substructure constraint $S$, we let $V(S,G)$ denote a vertex set, which contains all the vertices in $G$ that satisfy the substructure constraint $S$.
	
	Besides, a substructure constraint $S$ can be expressed by a SPARQL\cite{Zou2011gStore} query. For instance, $S_0$~(Figure~\ref{fig:subgraph_example}(b)) can be modeled as `$Select\ ?x$ $where\{$$?x$ $\langle fiendOf \rangle$ $v_3.$ $v_3\langle likes \rangle$$?y.\}$''.
	
\section{Uninformed Search}\label{section:uninformed_search} 
	
	The most straight-forward way of answering reachability queries on KGs is the uninformed~(blind) search strategies~\cite{Russell1995AIM193191}: DFS and BFS, which process an input query with no additional information beyond that in problem definition. Simply using label constraint to reduce the search space, we can adopt DFS or BFS to solve LCR queries~\cite{Jin2010CLR}. 
	
	Unfortunately, LSCR queries are too complex, and neither DFS nor BFS can be utilized to such queries directly, as they never revisit the passed vertices in the query processing. For example, supposing a label constraint $L$$=$$\{likes$, $hates, friendOf\}$, in order to verify the existence of $v_3\negthickspace\stackrel{L,S_0}{\leadsto}\negthickspace v_4$ in Figure~\ref{fig:subgraph_example}(a), a search algorithm needs to walk on the path $p\negthickspace=<\negthickspace v_3,likes,v_4,hates,v_1,friendOf,v_3,likes,v_4\negthickspace>$, because only $v_1$ and $v_2$ could satisfy $S_0$.

	To process a LSCR query $Q\negthickspace=\negthickspace(s,t,L,S)$ on a KG $G\negthickspace=\negthickspace(V$, $E, \mathcal{L}, L_{S})$, at least two procedures are required in DFS or BFS. One searches in the space whose vertices can be reached by $s$ under the label constraint $L$; the another one is executed only when the former discovers a vertex $v$ that satisfies the substructure constraint $S$, and runs from $v$ to $t$ by walking on the vertices that can be reached by $v$ under the label constraint $L$. This process stops, when the former stops, or the latter returns true.

	Then, we analyze the time complexity of the above process, assuming a substructure constraint $S\negthickspace=\negthickspace(?x, V_S, E_S, E_?)$. The former procedure requires $O\big(|V|(|V_S|+|E_S|+|E_?|)+|E|\big)$, as it runs in the space with evaluating whether each passed vertex can satisfy $S$ or not. Each invocation of the latter procedure needs $O(|V|\negthinspace+\negthinspace|E|)$. We execute the latter up to $|V(S,G)|$ times, which is less than $|V|$. Hence, we summarize the time complexity in Theorem~\ref{theorem:dfs_complexity}.

	\begin{theorem}
		\label{theorem:dfs_complexity}
		The time complexity of applying DFS/BFS is $O\big(|V|\negthinspace\times\negthinspace(|V_S|\negthinspace+\negthinspace|E_S|\negthinspace+\negthinspace|E_?|\negthinspace+\negthinspace|V|\negthinspace+\negthinspace|E|)\big)= O(|V|\negthinspace\times\negthinspace(|V|\negthinspace+\negthinspace|E|))$.
	\end{theorem}
	\begin{proof}
		It can be easily proved by the above discussion.
	\end{proof}

	Due to the high time complexity of utilizing the above techniques, we introduce a novel \underline{u}n\underline{i}nformed online \underline{s}earch strategy~(UIS, Algorithm~\ref{algorithm:UIS}) for LSCR queries~(Section~\ref{section:dfs/bfs:UIS}), which could be scaled on general labeled graphs, and provides a baseline for LSCR queries. After that, we discuss the usability of adopting current LCR methods on our problem~(Section~\ref{section:onlineSearch:LCR}).
	
	\begin{algorithm}[t]
		\caption{UIS Algorithm($G$, $Q$)}
		\label{algorithm:UIS}
		
		\begin{algorithmic}[1]
			\renewcommand{\algorithmicrequire}{\textbf{Input:}}
			\renewcommand{\algorithmicensure}{\textbf{Output:}}
			\renewcommand{\algorithmiccomment}[1]{\  #1}
			
			\REQUIRE $G=(V, E, \mathcal{L}, L_{S})$ represents a KG, \\$Q=(s,t,L,S)$ is a LSCR query.
			\ENSURE  The answer of Q.
			
			\STATE Let $\mathbb{S}$ be a stack with one element $s$ \label{line:dfs:initial_stack}
			
			\COMMENT{// The initial values in $close$ are set to N}
			\STATE $close[s]\leftarrow\,$SCck$(s,S)$ \label{line:dfs:close_initial}
			
			\STATE \textbf{while} $\mathbb{S}$ is not empty \textbf{do}\label{line:dfs:while_begin}
			\STATE \quad Take an element $u$ from $\mathbb{S}$ \label{line:dfs:pop_pair}
			\STATE \quad \textbf{for} each edge $e=(u,l,v)$, $l\in L$, incident to $u$ \textbf{do} \label{line:dfs:for_begin}
			
			\COMMENT{ \quad// We explore $v$ in the following cases}
			
			\STATE \quad \quad \textbf{if} $close[u]=T\wedge close[v]\neq T$ (\underline{case 1}) \textbf{then} \label{line:dfs:case1}
			\STATE \quad \quad \quad \quad Add $v$ into $\mathbb{S}$, $close[v]\leftarrow T$ \label{line:dfs:case1_add}
			\STATE \quad \quad \textbf{else if} $close[v]=N$ (\underline{case 2}) \textbf{then} \label{line:dfs:case2}
			\STATE \quad \quad \quad \quad Add $v$ into $\mathbb{S}$, $close[v]\leftarrow\,$SCck$(v,S)$ \label{line:dfs:case2_add}
			
			\STATE \quad \quad \textbf{if} $v=t\wedge close[v]=T$ \textbf{then} 
			\STATE \quad \quad \quad \textbf{return} $Q=T$  \label{line:dfs:end_loop} \label{line:dfs:end_for}
			\STATE \textbf{return} $Q=F$ \label{line:dfs:Q_F}
		\end{algorithmic}
	\end{algorithm}

\subsection{Baseline of LSCR Queries} \label{section:dfs/bfs:UIS}
	The algorithm UIS($G,Q$) starts with putting an element $s$ into a stack  $\mathbb{S}$~(Line~\ref{line:dfs:initial_stack}), and runs until $\mathbb{S}$ is empty~($s\negthickspace \stackrel{L,S}{\nrightarrow}\negthickspace t$, Line~\ref{line:dfs:Q_F}) or $s\negthickspace \stackrel{L,S}{\leadsto}\negthickspace t$ is proved~(Line~\ref{line:dfs:end_loop}). In UIS, we devise a function SCck$(v,S)$ to verify whether $v$ can satisfy $S$. Furthermore, in order to manage the recalling process of UIS, we develop a surjection $close\negmedspace:\negmedspace V$$\rightarrow$$\{N,T,F\}$ as illustrated in Definition~\ref{definition:close}.
	
	\begin{definition}
		\label{definition:close}
		A surjection $close\negmedspace:\negmedspace V$$\rightarrow$$\{N,T,F\}$ in a search algorithm for LSCR queries is as follows: 
		(\romannumeral1) initially, $\forall u\in V$, $close[u]$$=$$N$, which represents that $u$ have not been explored; (\romannumeral2) if $s\negthickspace\stackrel{L,S}{\leadsto}\negthickspace u$ and $s\negthickspace\stackrel{L}{\leadsto}\negthickspace u$ have been proved by the search strategy, we set $close[u]$ to $T$ and $F$, respectively.
	\end{definition}
	Note that, this surjection is used throughout this paper.

	After setting $close[s]$ to SCck$(s,S)$~(Line~\ref{line:dfs:close_initial}), the loop (Lines \ref{line:dfs:for_begin}-\ref{line:dfs:end_for}) takes an element $u$ from $\mathbb{S}$~(Line~\ref{line:dfs:pop_pair}), and processes each edge $e$$=$$(u,l,v)$ with the edge label $l$ belonging to $L$. Thus, the algorithm only traverses in the space that $s$ reaches to under $L$, i.e. $s\negmedspace\stackrel{L}{\leadsto}\negmedspace u$ exists. We explore $v$ in the following cases~(Lines~\ref{line:dfs:case1}-\ref{line:dfs:case2_add}). (1) $close[u]$$=$$T\wedge close[v]$$\neq$$T$, which means that $s\negthickspace\stackrel{L,S}{\leadsto}\negthickspace u$ exists. As $u\negmedspace\stackrel{L}{\leadsto}\negmedspace v$ and $e$$=$$(u,l,v)\wedge l\negthickspace\in\negthickspace L$, $s\negthickspace\stackrel{L,S}{\leadsto}\negthickspace v$ exists so that $close[v]$$=$$T$. (2) $close[u]$$\neq$$T\wedge close[v]$$=$$N$, where, if SCck$(v,S)$ returns true, i.e. $v$ satisfies $S$, we set $close[v]$$=$$T$~($s\negmedspace\stackrel{L,S}{\leadsto}\negmedspace v$), otherwise, we set $close[v]$$=$$F$.
	
	\textbf{Analysis.} In order to analyze the proposed algorithm, we first introduce the search tree in LSCR queries.
		
	For general consideration, in the following discussion, we assume SCck($s$,$S$) returns $false$, and loop (Line~\ref{line:dfs:while_begin}-\ref{line:dfs:end_loop}) starts with $close[s]$$=$$F$. Figure~\ref{fig:search_tree} draws a search tree example, in which, $\forall v\in V$, we make a distinction between $close[v]=T$ and $close[v]=F$, and mark them with colors.
	
	In order to illustrate the vertex exploring process, in Lines~\ref{line:dfs:case1}-\ref{line:dfs:case2_add}, we mark three edges with numbers \textit{1}, \textit{2} and \textit{3}, which denote $e_1$$=$$(u,l_1,v_1)$, $e_2$$=$$(u,l_2,v_2)$ and $e_3$$=$$(u',l_3,v_3)$, respectively. $e_1$ and $e_2$ appear, iff, in case~2, SCck$(v_1,S)$$=$$true$ and SCck$(v_2,S)$$=$$false$, respectively;  $e_3$ appears, iff, in case~1, we explore vertex $v_3$. In contrast, the edges in the search tree can only be generated in the above three ways, because we only explore a vertex in case~1 and case~2. We formally define the search tree in Definition~\ref{definition:search_tree}.
	\begin{definition}
		\label{definition:search_tree}
		A search tree of a LSCR query $Q$$=$$(s,t,L,S)$ is formalized as $\mathbb{T}$, where (\romannumeral1) the root of $\mathbb{T}$ is $s$; (\romannumeral2) supposing $N^\mathbb{T}$ denotes the node set of $\mathbb{T}$, there is an injection $f\negmedspace:\negmedspace V$$\rightarrow$$N^\mathbb{T}$ and, $\forall u\in V$, $u$ points at most two nodes $v^F$ and $v^T$ in $N^\mathbb{T}$; (\romannumeral3) $v^F$~($close[v]$$=$$F$) and $v^T$~($close[v]$$=$$T$) represent the states of $v$ in $close$ when vertex $v$ is passed by a searched algorithm.
	\end{definition}
	\begin{figure}[t]
		\centerline{\includegraphics[scale=0.37]{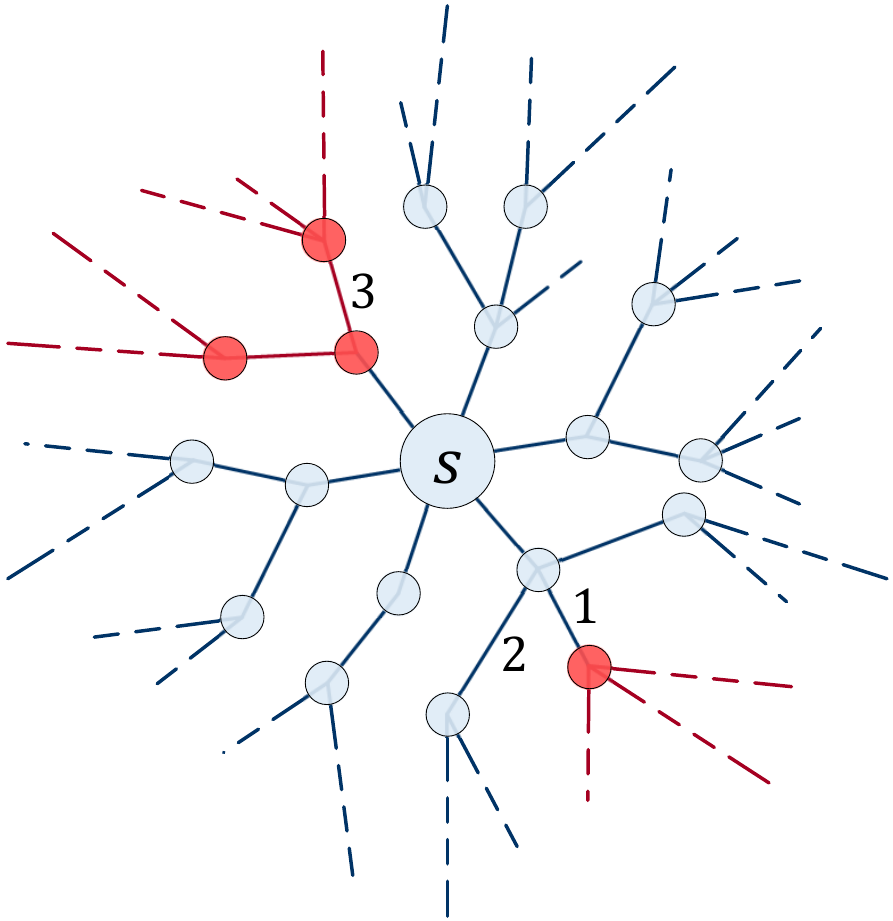}}
		\caption{An UIS search tree, where we make a distinction between $close[v]=T$~(red) and $close[v]=F$~(blue) by colors.} 
		\label{fig:search_tree}
	\end{figure} 
	
	Finally, we prove the algorithm correctness and the time complexity in Theorem~\ref{theorem:UIS} and Theorem~\ref{theorem:UIS_complexity}, respectively.
	
	\begin{theorem}
		\label{theorem:UIS}
		(Correctness of Algorithm~\ref{algorithm:UIS}.) $Q$ is a true query, iff, UIS($G$,$Q$) returns true.
	\end{theorem}
	\begin{proof}
		($\Rightarrow$) $Q$ is a true query, iff, $\exists p\negmedspace\in\negmedspace P(s,t)\wedge \exists v\negmedspace\in\negmedspace p$ and $v$ satisfies the substructure constraint $S$. Supposing $p$$=$$<s, e_0, \dots, v, \dots, e_k, t>$ and $v$ is the first vertex in $p$, where $s\negmedspace\stackrel{L}{\leadsto}\negmedspace v$, $v\negmedspace\stackrel{L}{\leadsto}\negmedspace t$ and $close[v]$$=$$T$ exist. As UIS has the ability of recalling a vertex $w$, where the state of $w$ in $close$ is $F$, the statement is obvious valid.
		
		($\Leftarrow$) If UIS returns true, a node $v^T$ exists in $\mathbb{T}$. Thus, $s\negthickspace\stackrel{L,S}{\leadsto}\negthickspace t$, and $Q$ is a true query.
	\end{proof}
	
	\begin{theorem}
		\label{theorem:UIS_complexity}
		The time complexity of UIS($G$,$Q$) is $O\big(|V|\times(|V_S|\negmedspace+\negmedspace|E_S|\negmedspace+\negmedspace|E_?|)+ |E|\big)$, assuming $S$$=$$(?x, V_S, E_S, E_?)$.
	\end{theorem}
	\begin{proof} 
		The time complexity of function $SCck(v,S)$ is $O(|V_S|\negmedspace+\negmedspace|E_S|\negmedspace+\negmedspace|E_?|)$. We invoke such function up to $|V|$ times. Plus, UIS traverses $G$ at most twice, according to Definition~\ref{definition:search_tree}. Hence, the time complexity of UIS traverses $G$ is $O(|V|+|E|)$. Thus, the statement is proved.
	\end{proof}
	
	\subsection{Analysis of LCR Methods} \label{section:onlineSearch:LCR}
	In this section, we first give a brief review of current LCR methods, then discuss why such methods cannot be applied on our problem.
	
	\textbf{Review.} The problem of LCR queries is originated in~\cite{Jin2010CLR}. They aim to address the high space complexity problem in the graph full transitive closure~(TC) precomputation, with a tree-based index framework. This framework contains a spanning tree~(or forest) and a partial transitive closure of the input graph, so that it can cover all the information of full TC to response the LCR queries.
	
	\cite{ZouXY0XZ14} decomposes an input graph into several strongly connected components. For each component a local transitive closure is computed, that is, they precompute all the CMSs of every vertex pair in a local transitive closure. 
	Then, they transfer each component into a bipartite graph. For the vertex pairs $(u, v)$ that $u$ and $v$ are not in the same component, both the CMSs from $u$ to $v$ and from $v$ to $u$ can be found in the topological order of bipartite graph union. That is, the CMSs of the vertices in different components are also precomputed. Thus, the precomputed index contains all the information of the full TC.
	
	To our knowledge, the current state-of-the-art solution of LCR queries is \cite{Valstar17Landmark}.
	Given a label set $\mathcal{L}$ of an input edge-labeled graph, they first choose $k$ landmarks. For each landmark vertex $v$, they precompute all CMSs from $v$ to every vertex that $v$ reaches.
	Besides, for accelerating LCR query processing, they also index each non-landmark vertex with $b$ CMSs. Plus, let $R_L(v)$$=$$\{w\negmedspace\in\negmedspace V| v \negmedspace\stackrel{L}{\leadsto}\negmedspace w\} $, $v$ is a landmark, and $\forall L \subseteq \mathcal{L},\ |L|\leq |\mathcal{L}|/4+1$, they precompute $R_L(v)$ for each landmark $v$, to accelerate false-queries.
	
	\textbf{Discussion.} Even though such techniques are fast on answering LCR queries, they cannot be scaled for indexing the relatively large KGs, due to the prohibitively high indexing time complexity.
	
	\begin{figure}
		\centering
		\renewcommand{\thesubfigure}{}
		\subfigure[(a)$|V|=5000$]{
			\includegraphics[scale = 0.6]{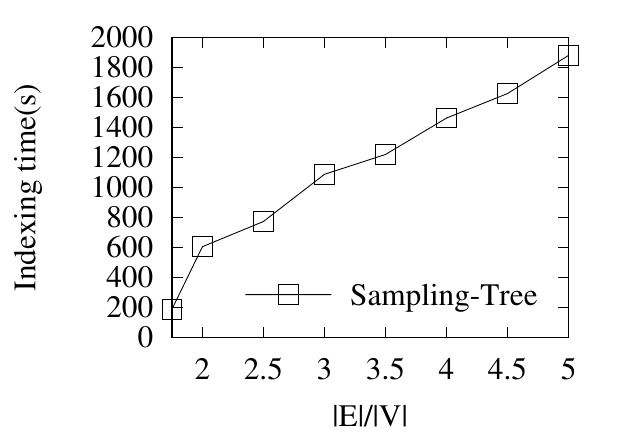}}
		\subfigure[(b)$D=1.5$]{
			\includegraphics[scale = 0.6]{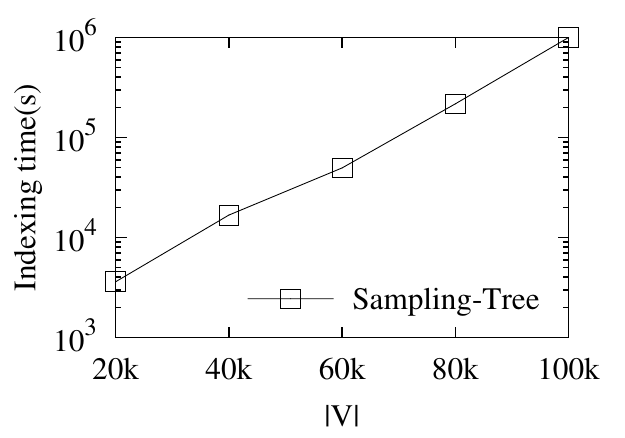}}
		\caption{Sampling-Tree indexing time, where $V$ is the vertex set of the indexed graph, and $D=|E|/|V|$ represents the graph density.}
		\label{fig:J_lcr_indexing}
	\end{figure}
	
	Based on the report in~\cite{Jin2010CLR}~(Table~1 of \cite{Jin2010CLR} and Table~2 of \cite{Jin2010CLR}), the indexing time of spanning tree increases linearly with the increasing of the graph density $D=|E|/|V|$, as depicted in Figure~\ref{fig:J_lcr_indexing}(a). Plus, the indexing time of spanning tree exponentially grows as $|V|$ grows, which is demonstrated in Figure~\ref{fig:J_lcr_indexing}(b). When $|V|=100,000$ and $D=1.5$, the indexing time is about $10^6$s$\approx$$11.57$d. 
	
	\cite{ZouXY0XZ14} represents a baseline of \cite{Valstar17Landmark}, and, in \cite{Valstar17Landmark}, they illustrate that \cite{ZouXY0XZ14} do not scale well on large graphs~($|V|>5.4$k). Besides, the indexing time complexity of \cite{Valstar17Landmark} is $O\big((|V|\log |V|+|E|+2^{|\mathcal{L}|}k+b(|V|-k))|V|2^{|\mathcal{L}|}\big)$, where $k$ is the number of selected landmarks, $b$ is the value of the indexed entries for each non-landmark vertex. In the experimental settings of \cite{Valstar17Landmark}, $k$$=$$1250+\negmedspace\sqrt{|V|}$ and $b$$=$$20$, so the time complexity of \cite{Valstar17Landmark} is equal to $O(|E||V|2^{|L|}+|V|^22^{2|L|})$. 
	
	Therefore, for a relatively large graph, the above methods can hardly be adopted.
	
\section{Improved Uninformed Search}\label{section:improved_u_s}
	
	For a LSCR query $Q$$=$$(s,t,L,S)$ on a KG $G$, UIS has to evaluate each  vertex of each search path $p$ with the function SCck, until $\exists v\negmedspace\in\negmedspace \mathbb{V}(p)\wedge v\negmedspace\in\negmedspace V(S,G)$ exists~($v$ satisfies the substructure constraint $S$). The time complexity of the above operation highly depends on the given substructure constraint $S$. Plus, UIS could not be accelerated by implementing the techniques for LCR queries, as mentioned in Section~\ref{section:onlineSearch:LCR}.
	
	As substructure constraints can be formatted as SPARQL queries~(Section~\ref{section:background}), we could obtain $V(S,G)$ by implementing SPARQL engines~\cite{SurveyOfSPARQLEngines} rather than checking the vertices in $G$ one by one. Then, a LSCR query $Q$ can be solved by recursively verifying the existence of $\exists v\in V(S,G),\ s\negthickspace\stackrel{L}{\leadsto}\negthickspace v\wedge v \negthickspace\stackrel{L}{\leadsto}\negthickspace t$. For example, in Figure~\ref{fig:subgraph_example}, $V(S_0,G_0)$$=$$\{v_1,v_2\}$. Then, for a LSCR query $Q_0\negthickspace=\negthickspace(v_3,v_4,L_0\negthickspace=\negthickspace\{likes,hates,friendOf\},S_0)$, the query process can be done throughout the sequence of verifying $v_3\negmedspace\stackrel{L}{\leadsto}\negmedspace v_1$ and $v_1\negmedspace\stackrel{L}{\leadsto}\negmedspace v_4$.
	
	Thus, in this section, we present the intuition idea of utilizing SPARQL engines on our problem, named UIS$^*$~(Algorithm~\ref{algorithm:SQLCS}). The algorithm description and analysis of UIS$^*$ are demonstrated in Section~\ref{section:sq:algorithm} and in Section~\ref{section:sq:correctness}, respectively.
	
	Note that, in this paper, we treat the elements in $V(S,G)$ as disordered, because existing (top-k) SPARQL engines~\cite{topKSPARQLSara,topKSPARQLShima} only can order the matched subgraphs with their own properties.
		
	\begin{algorithm}[t]
		\caption{UIS$^*$ Algorithm($G$, $Q$)}
		\label{algorithm:SQLCS}
		\begin{algorithmic}[1]
			\renewcommand{\algorithmicrequire}{\textbf{Input:}}
			\renewcommand{\algorithmicensure}{\textbf{Output:}}
			\renewcommand{\algorithmiccomment}[1]{\  #1}
			
			\REQUIRE $G=(V, E, \mathcal{L}, L_{S})$ represents a KG, \\$Q=(s,t,L,S)$ is a LSCR query,\\ $V(S,G)$ is obtained by a SPARQL engine.
			\ENSURE  The answer of Q.
			
			\STATE Let $\mathbb{S}$ be a global stack with an element $s$\label{line:sm:initial_stack} \label{line:sm:add_1}
			\STATE $close[s]\leftarrow F$\label{line:sm:initial_close} \COMMENT{// The initial values in $close$ are set to N}
			
			\STATE \textbf{for} each vertex $v$ in $V(S,G)$ \textbf{do} \label{line:sm:loop_start}
			\STATE \quad \textbf{if} $close[v]=N$ \textbf{then} \label{line:sm:no_v_start}
			\STATE \quad \quad \textbf{if} $v=s$ \textbf{or} $v=t$ \textbf{then}
			\STATE \quad \quad \quad \textbf{return} $Q=\,$LCS$(s,t,L,F)$ \label{line:sm:s_v_1}

			\STATE \quad \quad \textbf{else if} LCS$(s,v,L,F)$ \textbf{then} \label{line:sm:s_v_2}
			\STATE \quad \quad \quad \textbf{if} LCS$(v,t,L,T)$ \textbf{then} \label{line:sm:v_t_1}
			\STATE \quad \quad \quad \quad \textbf{return} $Q=T$ \label{line:sm:no_v_end}
			\STATE \quad \textbf{else if} $close[v]=F$ \textbf{then} \label{line:sm:have_v_start}
			\STATE \quad \quad \textbf{if} LCS$(v,t,L,T)$ \textbf{then}\label{line:sm:v_t_2}
			\STATE \quad \quad \quad \textbf{return} $Q=T$ \label{line:sm:loop_end} \label{line:sm:have_v_end}
			
			\STATE \textbf{return} $Q=F$ 
			
			\rule[2pt]{7.84cm}{0.02em}
			
			\STATE	\textbf{Function:} LCS$(s^*, t^*, L, B)$ \COMMENT{// B is a boolean}\label{line:sm:lcs_start}
			\STATE \quad \textbf{if} $B=T$ \textbf{then} \label{line:sm:lcs:B_T_add_stack}
			\STATE \quad \quad $close[s^*]\leftarrow T$, add $s^*$ into $\mathbb{S}$ \label{line:sm:add_2}
			\STATE \quad \textbf{while} $B=F\wedge \mathbb{S}\neq \phi$ \textbf{or} $B=close[\mathbb{S}.first]=T$ \textbf{do} \label{line:sm:lcs:loop_start}
			\STATE \quad \quad Take an element $u$ from $\mathbb{S}$ \label{line:sm:lcs:takeout}
			\STATE \quad \quad \textbf{for} each edge $e=(u,l,w)$, $l\in L$, incident to $u$ \textbf{do} \label{line:sm:lcs:for}
			\COMMENT{ \quad \quad// We explore $w$ in the following cases}
			\STATE \quad \quad \quad \textbf{if} $B=T\wedge close[w]\neq T$ (\underline{case 1}) \textbf{or}\\
			\quad \quad \quad $B=F\wedge close[w]=N$ (\underline{case 2}) \textbf{then}	\label{line:sm:lcs:cases_new}		
			\STATE \quad \quad \quad \quad  Add $w$ into $\mathbb{S}$, $close[w]\leftarrow B$ \label{line:sm:lcs:loop_end} \label{line:sm:add_4}
			\STATE \quad \quad \quad \textbf{if} $w=t^*$ \textbf{then}
			\STATE \quad \quad \quad \quad \textbf{return true}\label{line:sm:lcs:stproved}
			\STATE \quad Remove $x$ from $\mathbb{S}$, if $close[x]=T$, then \textbf{return false} \label{line:sm:lcs:function_end} 
		\end{algorithmic}
	\end{algorithm}	
	
	\subsection{Algorithm Description} \label{section:sq:algorithm}
	
	In Algorithm~\ref{algorithm:SQLCS}, a global stack $\mathbb{S}$ is initialized with an element $s$~(Line~\ref{line:sm:initial_stack}), and used throughout the algorithm. The surjection $close$ is introduced in Definition~\ref{definition:close}. To be notice, the definition of the UIS$^*$ search tree is same to Difinition~\ref{definition:search_tree}, the correctness of which is proved in Lemma~\ref{lemma:tree}.
	
	Loop~(Lines \ref{line:sm:loop_start}-\ref{line:sm:loop_end}) processes each vertex $v\in V(S,G)$, with a function LCS$(s^*,t^*,L,B)$~(Lines \ref{line:sm:lcs_start}-\ref{line:sm:lcs:function_end}) to evaluate the reachability from a vertex $s^*$ to a vertex $t^*$ under a label constraint $L$. The value of parameter $B$ is in the range of $\{T,F\}$, where, if $s^*\negthickspace\in\negmedspace V(S,G)$, $B$$=$$T$, otherwise, $B$$=$$F$. For clarity, we draw two sequences of the search trees with an assumption~($g,h\negmedspace\in\negmedspace V(S,G)\wedge g\negmedspace\stackrel{L}{\nrightarrow}\negmedspace t$), which are shown in Figure~\ref{fig:sqlcs_search_tree_normal} and Figure~\ref{fig:sqlcs_search_tree_special}, respectively.
	
	For general consideration, supposing $s\notin V(S,G)$ in the demonstration of the UIS$^*$ running process. Thus, in the first iteration of the loop, $\forall v\in V(S,G)$, $close[v]=N$. We process $v$ in Lines~\ref{line:sm:no_v_start}-\ref{line:sm:no_v_end}, to verify the existence of $s\stackrel{L}{\leadsto}v$~(Line~\ref{line:sm:s_v_1} and Line~\ref{line:sm:s_v_2}). Such that, the function LCS is invoked with the parameters $s^*\negthickspace=\negmedspace s$, $t^*\negthickspace=\negmedspace v$, $L\negmedspace=\negmedspace L$ and $B\negmedspace=\negmedspace F$~($s^*\negthickspace=\negmedspace s\notin\negmedspace V(S,G)$). Since $B\negmedspace=\negmedspace F$, function LCS only explores the vertices $w$ with $close[w]\negmedspace=\negmedspace N$~(\textit{case 2}, Line~\ref{line:sm:lcs:cases_new}) and change the state of the explored vertices in $close$ to $F$~(Line~\ref{line:sm:add_4}). In the following illustration, we discuss UIS$^*$ by the results of the first invocation of LCS, where ``true'' and ``false'' relate to $s\stackrel{L}{\leadsto}v$ and $s\stackrel{L}{\nrightarrow}v$, respectively.

	\underline{$s\stackrel{L}{\leadsto}v$.} Assuming $v\negmedspace=\negmedspace g$, Figure~\ref{fig:sqlcs_search_tree_normal}(a) depicts the search tree after the first invocation of LCS. As $s\negmedspace\stackrel{L}{\leadsto}\negmedspace g$ is proved, we start the function LCS, and set the parameters $s^*,t^*,L$ and $B$ to $g,t,L$ and $T$~($s^*\negthickspace=\negmedspace g\negmedspace\in\negmedspace V(S,G)$), respectively, in Line~\ref{line:sm:v_t_1}. Then, with $B$$=$$T$~($s\negmedspace\stackrel{L}{\leadsto}\negmedspace s^*$), we add an element $s^*$ into $\mathbb{S}$ and set $close[s^*]$ to $T$~(Line~\ref{line:sm:lcs:B_T_add_stack}). Then, for each edge $e=(u,l,w)$ that is incident to a taken out element $u$~(Line~\ref{line:sm:lcs:takeout}) with $l\in L$, we additionally explore $w$ in \textit{case~1}~(Line~\ref{line:sm:lcs:cases_new}) with changing $close[w]$ to $B$~(Line~\ref{line:sm:add_4}). 
	
	Since the implemented data structure $\mathbb{S}$ is a stack, in which the elements obey the LIFO~(last-in-first-out) principle, the new added elements in this invocation of LCS are processed first. Certainly, in the second invocation of LCS, if an element $x$ in $\mathbb{S}$ and $close[x]\negmedspace=\negmedspace T$, then $s^*\negthickspace\stackrel{L}{\leadsto}\negmedspace x$, while, if $close[x]\negmedspace=\negmedspace F$, the existence of  $s^*\negthickspace\stackrel{L}{\leadsto}\negmedspace x$ is unknown.
		
	Thus, the loop in LCS runs until $s^*\negmedspace\stackrel{L}{\leadsto}t^*$ is proved, or the state of first element in $\mathbb{S}$ is $F$~($close[\mathbb{S}.first]\neq T$), where: \textbf{(\romannumeral1)} In the former condition, as $s\negmedspace\stackrel{L}{\leadsto}\negmedspace s^*$, $s^*\negthickspace\in\negmedspace V(S,G)$, $s^*\negthickspace\stackrel{L}{\leadsto}\negmedspace t^*$ and $t^*\negmedspace=t$ all exist, we have $s\negmedspace\stackrel{L}{\leadsto}\negmedspace t$ in the main function. Thus, UIS$^*$ stops and returns $Q\negmedspace=\negmedspace T$~(Line~\ref{line:sm:no_v_end} and Line~\ref{line:sm:have_v_end}). \textbf{(\romannumeral2)} In the latter condition, Line~\ref{line:sm:lcs:function_end} removes the elements in $\mathbb{S}$ that have been passed in this invocation~(whose states in $close$ are equal to $T$). For instance, in Figure~\ref{fig:sqlcs_search_tree_normal}(b), the elements $x$ and $y$ are removed from $\mathbb{S}$, as the node $x^T$ and the node $y^T$ are already in the search tree. Furthermore, Figure~\ref{fig:sqlcs_search_tree_normal}(c) describes a possible later sequence of Figure~\ref{fig:sqlcs_search_tree_normal}(b), as $h\negmedspace\in\negmedspace V(S,G)$.	
	
	\begin{figure}[t]
		\centering
		\renewcommand{\thesubfigure}{}
		\subfigure[(a)]{
			\includegraphics[scale = 0.37]{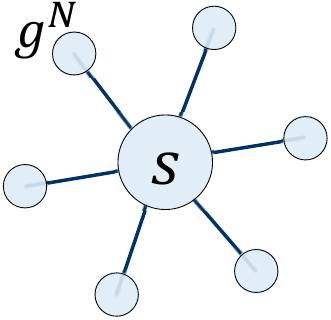}}
		\subfigure[(b)]{
			\includegraphics[scale = 0.37]{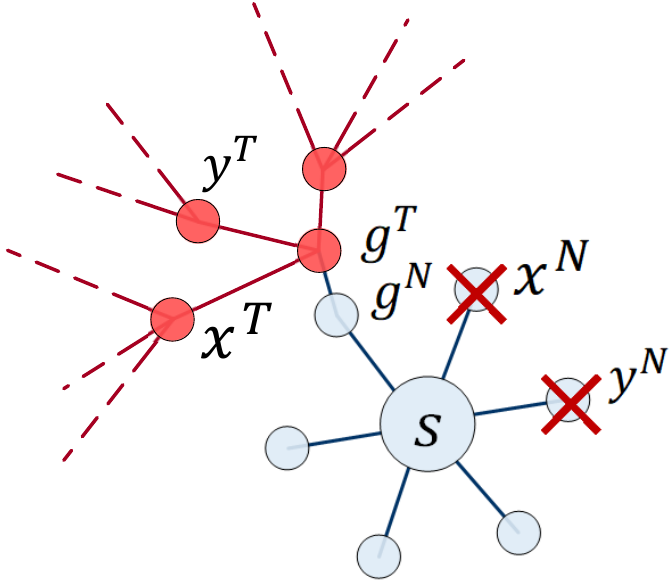}}
		\subfigure[(c)]{
			\includegraphics[scale = 0.37]{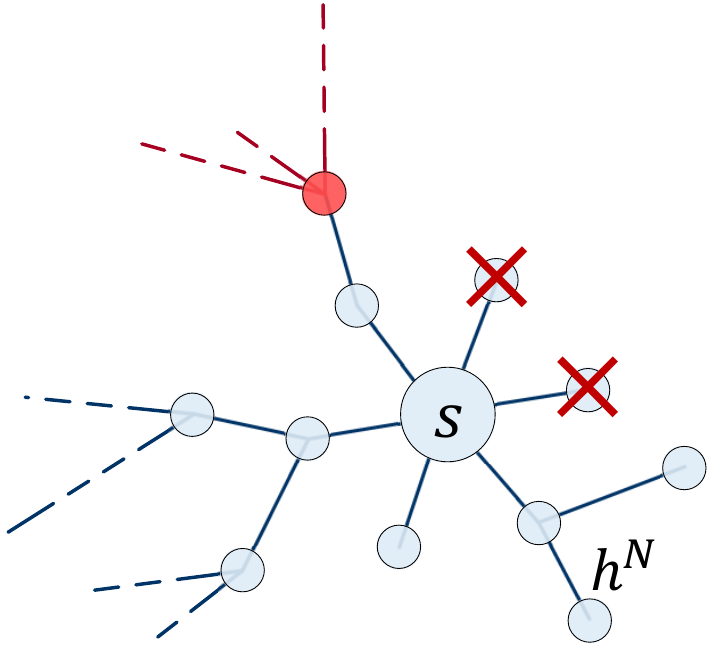}}
		\caption{A sequence of search trees generated by one execution of UIS$^*$.}
		\label{fig:sqlcs_search_tree_normal}
	\end{figure}
	
	\begin{figure}[t]
		\centering
		\renewcommand{\thesubfigure}{}
		\subfigure[(a)]{
			\includegraphics[scale = 0.37]{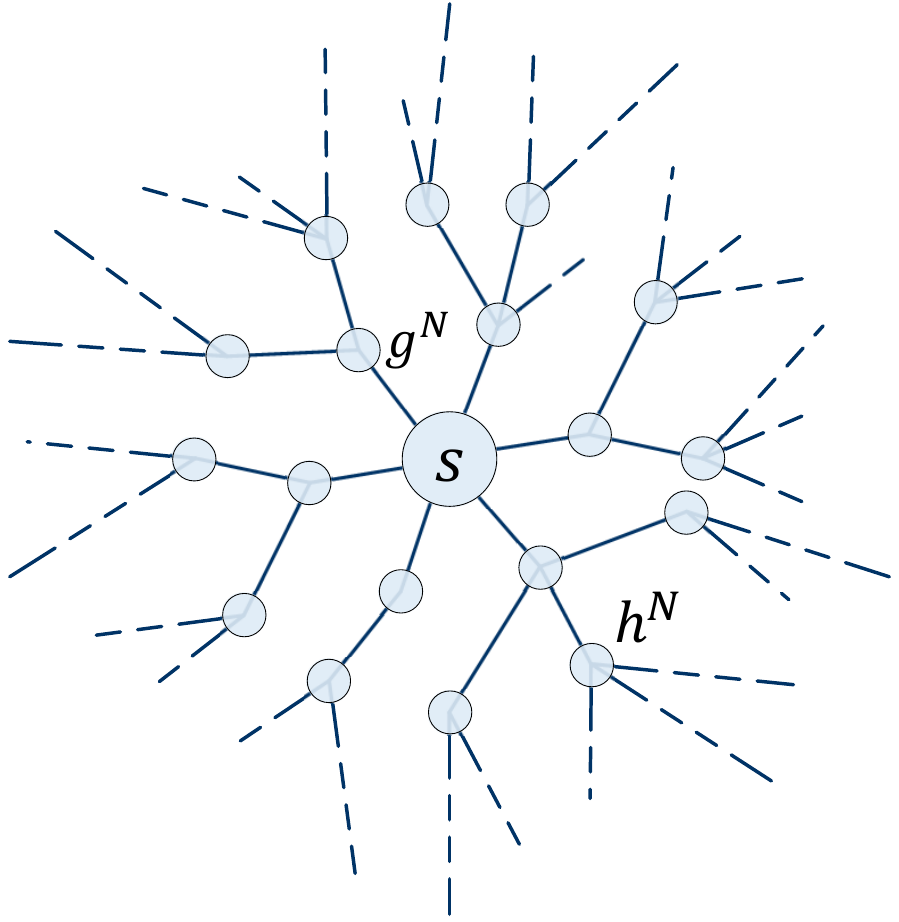}}
		\subfigure[(b)]{
			\includegraphics[scale = 0.37]{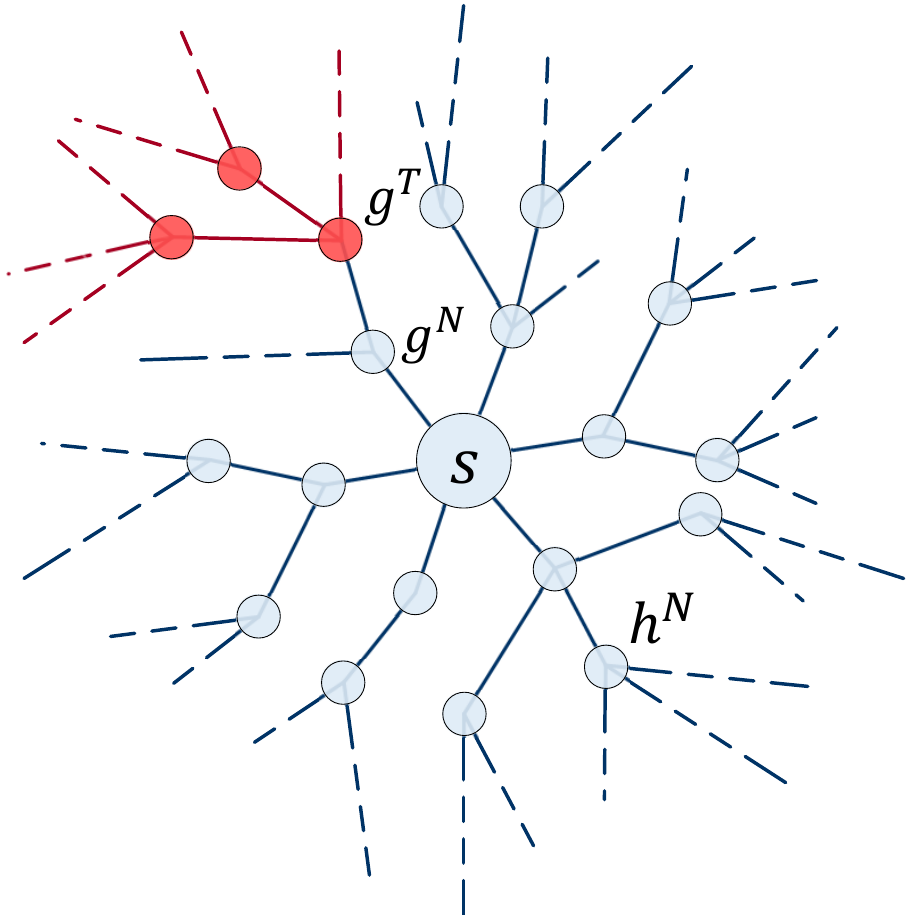}}
		\caption{The sequence is generated when the first invocation of LCS returns false.}
		\label{fig:sqlcs_search_tree_special}
	\end{figure}
	
	\underline{$s\stackrel{L}{\nrightarrow}v$.} Since $B\negmedspace=\negmedspace F$, after the first invocation of LCS, $\forall w\in V$, if $s\stackrel{L}{\leadsto}w$, $close[w]=F$, and, if $s\stackrel{L}{\nrightarrow}w$, $close[w]$$=$$N$, and the stack $\mathbb{S}$ is empty. In other words, UIS$^*$ has traversed all the vertices that $s$ reaches to. Figure~\ref{fig:sqlcs_search_tree_special}(a) demonstrates the search tree of UIS$^*$ at this time. In a later iteration of loop~(Lines \ref{line:sm:loop_start}-\ref{line:sm:loop_end}), only when $close[v]\negthickspace=\negthickspace F$ (Line~\ref{line:sm:have_v_start}), we evaluate the existence of  $v\stackrel{L}{\leadsto}t$~(Line~\ref{line:sm:v_t_2}). A schematic view of this sequence is depicted in Figure~\ref{fig:sqlcs_search_tree_special}(b). Based on the above analysis, we present a theorem to show that the order of processing the elements in $V(S,G)$ dominates the efficiency of UIS$^*$.
	
	\begin{theorem}
		\label{theorem:sqlcs_weakness}
		If, in the $j^{th}$ invocation of LCS with $B=F$, LCS returns false, then, after the $j^{th}$ invocation, $\forall w\in V\wedge s\stackrel{L}{\leadsto}w$, $close[w]\neq N$.
	\end{theorem}
	
	\begin{proof}
		That can be proved by the above discussion.
	\end{proof}

	%

	
	\subsection{Analysis} \label{section:sq:correctness}
	In this section, we analysis the correctness and time complexity of UIS$^*$($G$,$Q$), where $G=(V, E, \mathcal{L}, L_{S})$ is a KG and $Q=(s,t,L,S)$ is a LSCR query. Plus, we prove that the search tree of UIS$^*$ can be formalized as that of UIS in Lemma~\ref{lemma:tree}.

	\setcounter{lemma}{1}	
	\begin{lemma}
		\label{lemma:s_to_everynode}
		Assuming $s\negthickspace\stackrel{L}{\leadsto}\negthickspace s$, $\forall\negthinspace v\negthinspace\in\negthinspace V, s\negthickspace\stackrel{L}{\leadsto}\negthickspace v, iff\negthinspace, close[v]\negthinspace\neq\negthinspace N$.
	\end{lemma}
	\begin{proof}
		$\forall v\in V$, the value of $close[v]$ is initialized to $N$, which can only be changed in Line~\ref{line:sm:add_1} and function LCS~(Lines \ref{line:sm:lcs_start}-\ref{line:sm:lcs:function_end}). 
		Thus, with the assumption, we prove the statement by induction on the LCS invocation times:
		
		\textbf{First invocation.} According to the algorithm description~(Section~\ref{section:sq:algorithm}), loop~(Lines~\ref{line:sm:lcs:loop_start}-\ref{line:sm:lcs:loop_end}) maintains a loop invariant in this invocation that is same to the statement.  
		
		\textbf{k$^{th}$ invocation.} Assuming the statement is correct in the $(k-1)^{th}$ invocation, the above loop still holds the statement as a loop invariant in the k$^{th}$ invocation: Since $s\negmedspace\stackrel{L}{\leadsto}\negmedspace s^*$ must have been proved, the invariant exists before the first iteration of this loop. Plus, as we only explore the vertices in the mentioned two cases, the invariant is correct during the later iterations.
	\end{proof}

	Then, Lemma~\ref{lemma:tree} states that the search tree of UIS$^*$($G$,$Q$) is as defined in Definition~\ref{definition:search_tree}.
	
	\begin{lemma}
		\label{lemma:tree}
		All the search paths of  UIS$^*$ constitute \textbf{one search tree} $\mathbb{T}$, at any time of one execution, as we defined in Definition~\ref{definition:search_tree}.	
	\end{lemma}
	\begin{proof}
		This lemma is valid based on Lemma~\ref{lemma:s_to_everynode} and Definition~\ref{definition:close}.
	\end{proof}

	
	Finally, the correctness and time complexity of UIS$^*$ are illustrated in Theorem~\ref{theorem:SQLCS} and Theorem~\ref{theorem:SQLCS_complexity}, respectively. 
	\setcounter{theorem}{3}
	\begin{theorem}
		\label{theorem:SQLCS}
		$Q$ is a true query, iff, UIS$^*$ returns true.
	\end{theorem}
	\begin{proof}
		($\Rightarrow$) If $Q$ is a true query, $s \stackrel{L,S}{\leadsto} t$. Then, $\exists p\in P(s,t)$, $\mathbb{L}(p)\subseteq L$ and $\exists u\in \mathbb{V}(p)\wedge u\in V(S,G)$~(Theorem~\ref{theorem:LSCR}). Supposing $p$$=$$<s, e_0, \dots, v, \dots, e_k, t>$, $\exists v^T$ in $\mathbb{T}$ and $t^T$ in $\mathbb{T}$ both exist. Thus, UIS$^*$ returns true. ($\Leftarrow$) If UIS$^*$($G$,$Q$) returns true, a path exists in $\mathbb{T}$ that satisfies both the label constraint $L$ and the substructure constraint $S$. Thus, $Q$ is a true query.
	\end{proof}

	
	\begin{theorem}
		\label{theorem:SQLCS_complexity}
		The time complexity of UIS$^*$ is $O(|V|+|E|)$.
	\end{theorem}
	\begin{proof}
		Based on Lemma~\ref{lemma:tree}, UIS$^*$ processes each vertex in $G$ at most two times.
	\end{proof}
	
	\begin{figure}[t]
		\centering
		\renewcommand{\thesubfigure}{}
		\subfigure[(a) $u_0\rightarrow u_5$]{
			\includegraphics[scale = 0.37]{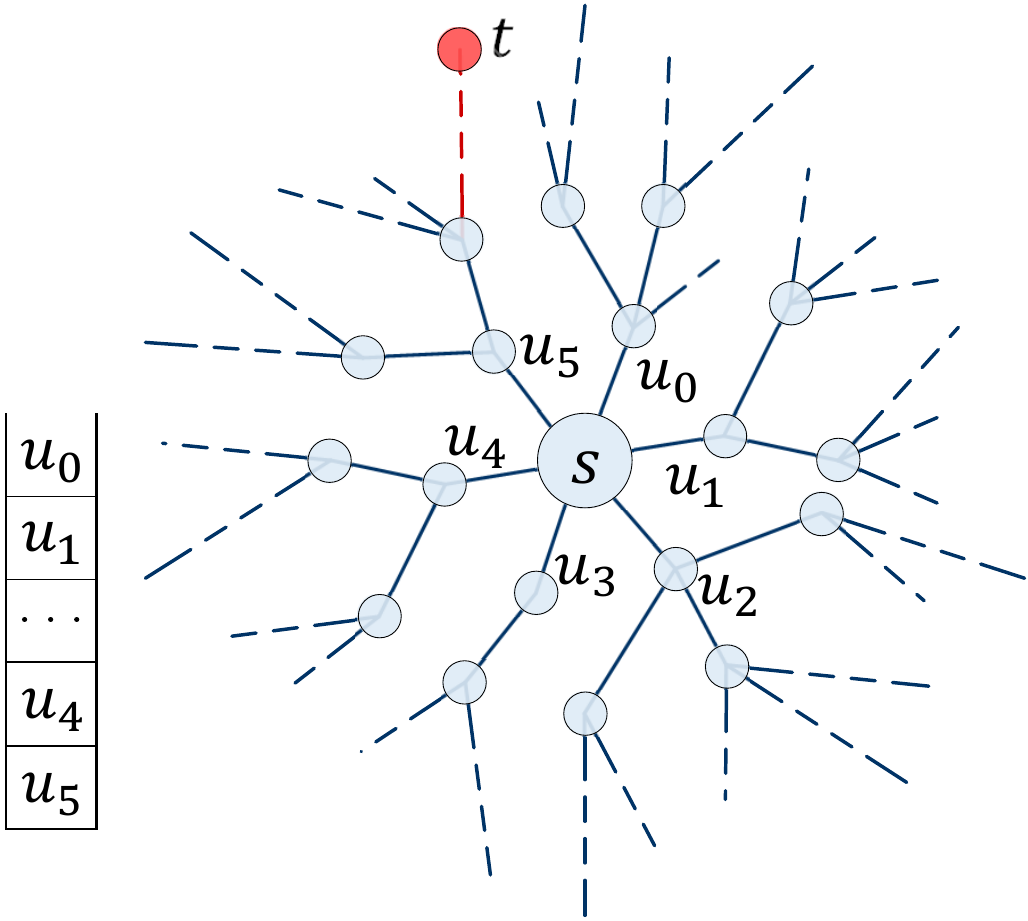}}
		\subfigure[(b) $u_5 \rightarrow u_0$]{
			\includegraphics[scale = 0.37]{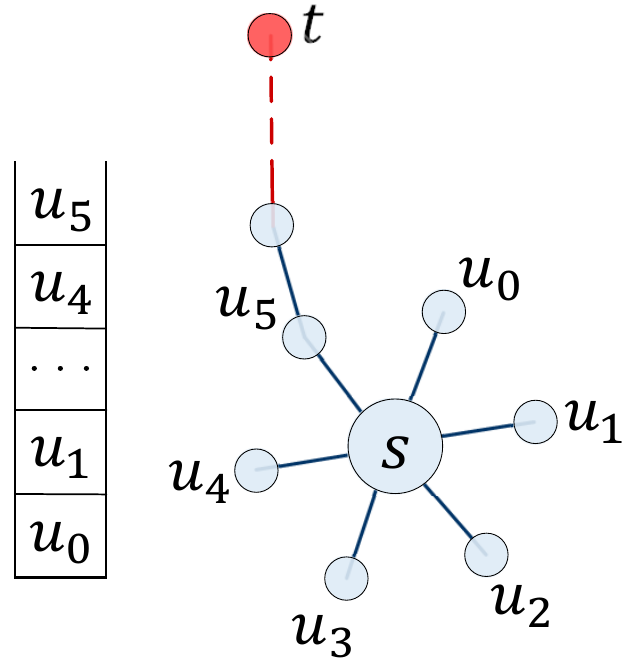}}
		\caption{Two search trees from $s$ to $t$ with different search directions.}
		\label{fig:search_direction}
	\end{figure}

\section{Informed Search}\label{section:informed_search}

	In order to address LSCR queries on KGs, two uninformed search strategies are introduced, UIS~(Agorithm~\ref{algorithm:UIS}) and UIS$^*$~(Algorithm~\ref{algorithm:SQLCS}). The former is a baseline method for LSCR queries; the latter obtains $V(S,G)$ by implementing SPARQL engines and has a lower time complexity.
	
	However, the uninformed search algorithms have their own limitations. Their query efficiencies are not only dependent on the sizes of the input KGs~(Theorem~\ref{theorem:UIS_complexity} and Theorem~\ref{theorem:SQLCS_complexity}), but also dominated by the search direction~(Theorem~\ref{theorem:sqlcs_weakness}). We present an illustration of the latter limitation with Figure~\ref{fig:search_direction}, as the above proposed strategies both obey the LIFO principle~(Section~\ref{section:uninformed_search} and Section~\ref{section:improved_u_s}). Assuming a vertex $s$ relates to six vertices~($u_0,\dots,v_5$, Figure~\ref{fig:search_direction}), in which only $u_5$ and a vertex $t$ are reachable. We draw two possible search trees of an LIFO algorithm from $s$ to  $t$. One search tree in Figure~\ref{fig:search_direction}(a) follows the order of $u_0\rightarrow u_5$, at the beginning of the query processing, while another in Figure~\ref{fig:search_direction}(b) starts by the order of $u_5\rightarrow u_0$. Obviously, the query efficiencies of the above two orders are beyond comparison.

	The traditional informed search algorithms~\cite{Russell1995AIM193191}, e.g. best-first search and A* search, explore the input graphs with an evaluation~(heuristic) function. With such function, an informed search method could break the fixed search orders of both LIFO and FIFO principles, and may find a relatively short search path to improve the query efficiency. Unfortunately, the above informed search strategies are no-recall, which is similar to DFS/BFS~(Section~\ref{section:uninformed_search}). Plus, the evaluation functions of the above techniques are too simple to be extended for the intricate LSCR queries on KGs.
	
	Inspired by the traditional informed search techniques, in this section, we propose an \underline{in}formed \underline{s}earch algorithm for LSCR queries on KGs, named INS. This algorithm is similar to UIS$^*$, except the following two points. Firstly, it utilizes a lightweight index~(\textit{local index}) in Section~\ref{section:libaq:indexing_strategy} to reduce the online computational consumption, within a bounded indexing time complexity. Secondly, INS applies two data structures~(Section~\ref{section:libaq:search}), a priority heap $\mathbb{H}$ and a priority queue $\mathbb{Q}$, to implement the evaluation function as that in the traditional informed strategies.
	
	\begin{figure}[t]
		\centering
		\renewcommand{\thesubfigure}{}
		\subfigure[(a) Traditional]{
			\includegraphics[scale = 0.35]{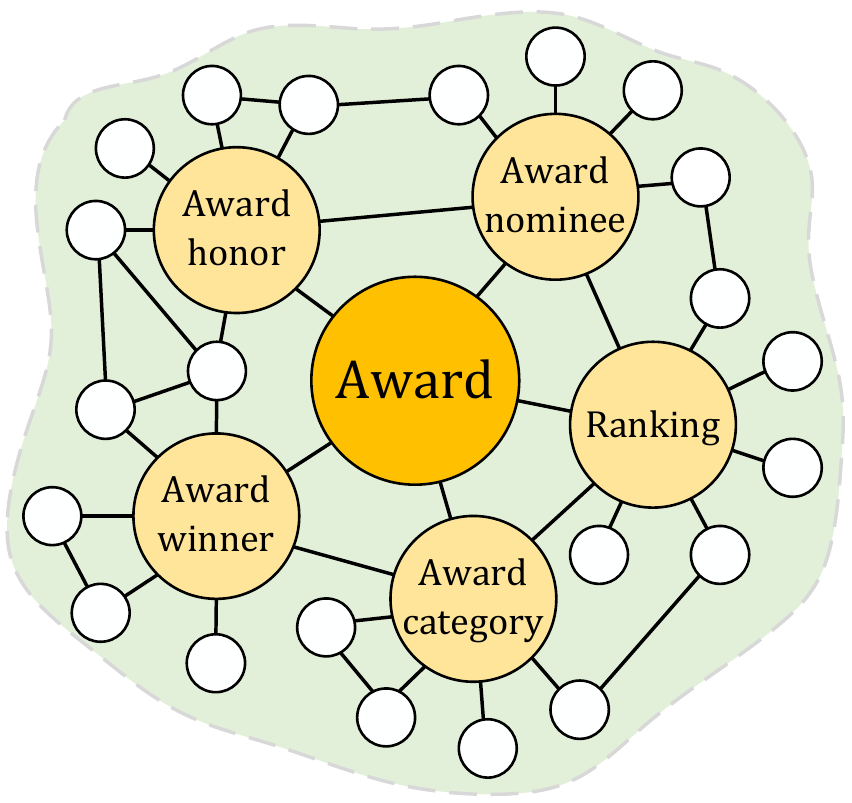}}
		\subfigure[(b) Local index]{
			\includegraphics[scale = 0.35]{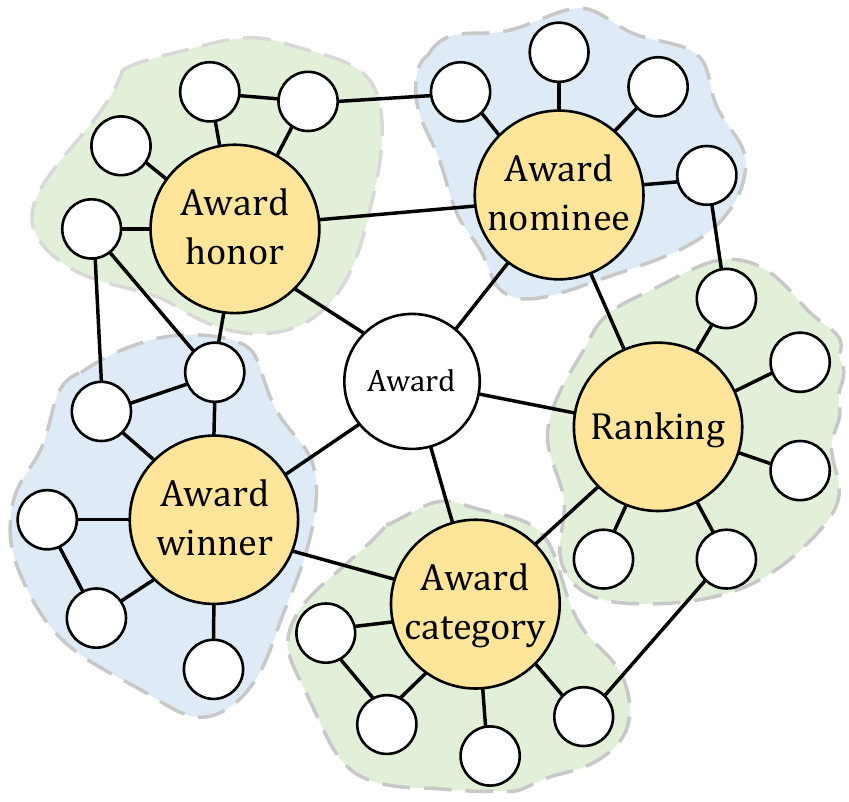}}
		\caption{The comparison of two indexing strategies on Freebase, where we highlight the chosen landmarks.}
		\label{fig:indexing_comparison}
	\end{figure}
		
	\subsection{Indexing Strategy}\label{section:libaq:indexing_strategy}
	
	The intuition idea of our \textit{local index} on a KG $G$$=$$(V, E, \mathcal{L}, L_{S})$ is similar to landmark indexing~\cite{Valstar17Landmark}~(Section~\ref{section:onlineSearch:LCR}). The difference is that we could bound the indexing time complexity, by only precomputing each chosen landmark in a specific subgraph of $G$, instead of in the whole input KG. 
	
	In the following part, we overview our indexing technique in Section~\ref{section:libaq:indexing:overview}, in which the \textit{local index} is formally defined. Then, the algorithm description of the indexing strategy are depicted in Section~\ref{section:libaq:indexing:description}, followed by the correctness and complexity analysis.
		
	\subsubsection{Overview}\label{section:libaq:indexing:overview}
	This overview starts from the illustration of the difference between \cite{Valstar17Landmark} and local index, with an example~(Figure~\ref{fig:indexing_comparison}), then presents a formal definition of the local index.
	
	According to the description of \cite{Valstar17Landmark} in Section~\ref{section:onlineSearch:LCR}, we formalize the traditional landmark indexing as a surjection $f\negmedspace:\negmedspace\mathcal{I}$$\rightarrow$$\{G\}$, where $\mathcal{I}$ denotes the set of the chosen landmarks~(highest degrees), and $\{G\}$ represents the index ranges of the landmarks in $\mathcal{I}$. For example, in Figure~\ref{fig:indexing_comparison}(a), $\mathcal{I}$$=$$\{$\textit{the highlighted vertices of Figure~\ref{fig:indexing_comparison}(a)}$\}$, and, for each vertex $v$ in $\mathcal{I}$, \cite{Valstar17Landmark} precomputes the CMSs from $v$ to any other vertices in Figure~\ref{fig:indexing_comparison}(a). Obviously, when the size of $G$ grows, the indexing time of such method grows exponentially, as we discussed in Section~\ref{section:onlineSearch:LCR}.
	
	To conquer the unbounded indexing time complexity, we aim to narrow the range of the precompution for each chosen landmark from the whole $G$ to a subgraph, as shown in Figure~\ref{fig:indexing_comparison}(b).  For formal description, a bijection $\mathcal{F}\negmedspace:\negmedspace\mathcal{I}$$\rightarrow$$\mathcal{G}$ exists in our indexing strategy~(Algorithm~\ref{algorithm:Indexing}), where $\mathcal{I}$ is the set of the chosen landmarks, and, for a landmark $u\in\mathcal{I}$, $\mathcal{F}(u)\in\mathcal{G}$ is the subgraph that the landmark $u$ belongs to. Importantly, for each non-landmark vertex $v$ in $\mathcal{F}(u)$, we stipulate that the vertex $u$ reaches the vertex $v$, and, for a subgraph $\mathcal{F}(u)\in\mathcal{G}$, $u$ is the only one chosen landmark. Plus, the result of combining all the subgraphs in $\mathcal{G}$ does not necessarily contain all elements of $G$.

	\begin{algorithm}[t]
		\caption{Indexing Algorithm($G$)}
		\label{algorithm:Indexing}
		\begin{algorithmic}[1]
			\renewcommand{\algorithmicrequire}{\textbf{Input:}}
			\renewcommand{\algorithmicensure}{\textbf{Output:}}
			\renewcommand{\algorithmiccomment}[1]{\  #1}
			
			\REQUIRE $G=(V, E, \mathcal{L}, L_{S})$ represents a KG,
			\ENSURE  Local index: $I\negthinspace I[u]\cup E\negthinspace I^T[u]\cup D[u]$ entries~($u\in\mathcal{I}$).
			
			\STATE $\mathcal{I}\leftarrow\,$LandmarkSelect($L_S,k$)\label{line:bijection:initial_I} \COMMENT{//$\,$Select k landmarks}
			\STATE BFSTraverse($\mathcal{I}$) \label{line:indexing:bijection}  \COMMENT{//$\,w.A_\mathcal{F}\negthickspace=\negthickspace u\Leftrightarrow \mathcal{F}(u)$ contains $w$}
			\STATE \textbf{for} vertex $u$ in $\mathcal{I}$ \textbf{do} \label{line:indexing:for_s} 
			\STATE \quad LocalFullIndex($u$) \label{line:indexing:for_e}
			
			\rule[2pt]{7.84cm}{0.02em}
			
			\STATE	\textbf{Function:} LocalFullIndex$(u)$ \label{line:indexing:function_starts}
			\STATE \quad $I\negthinspace I[u]$, $E\negthinspace I[u]$, $E\negthinspace I^T[u]\leftarrow\,$empty (key, value) pair sets\label{line:indexing:empty_sets}
			\STATE \quad Let $\mathbb{Q}$ be a queue with an element $(u,\{\})$ \label{line:indexing:queue}
			\STATE \quad \textbf{while} $\mathbb{Q}$ is not empty \textbf{do} \label{line:indexing:loop_starts}
			\STATE \quad \quad Take a pair ($v$,$L$) from $\mathbb{Q}$ \label{line:indexing:take_out_a_pair}
			\STATE \quad \quad \textbf{if} Insert($v,L,I\negthinspace I[u]$) = true \textbf{then} \label{line:indexing:case1_s}
			\STATE \quad \quad \quad \textbf{for} each edge $(v,l,w)$ incident to $v$ \textbf{do}\label{line:indexing:explore_s}
			\STATE \quad \quad \quad \quad \textbf{if} $w.A_\mathcal{F}=u$ \textbf{then}\label{line:indexing:sameLPFunction}
			\STATE \quad \quad \quad \quad \quad Add pair ($w,L\cup\{l\}$) into $\mathbb{Q}$\label{line:indexing:add_Q}
			\STATE \quad \quad \quad \quad \textbf{else} $\,$ Insert($w,L\cup\{l\},E\negthinspace I[u]$)\label{line:indexing:add_exteral} 
			\STATE \quad $E\negthinspace I^T[u], D[u]\leftarrow\,$Compute$(E\negthinspace I[u])$\label{line:indexing:function_new_ends}
			
			\rule[2pt]{7.84cm}{0.02em}
			
			\STATE \textbf{Function:} Insert$(v, L,index[u])$ \label{line:indexing:function:function_starts}
			\STATE \quad \textbf{if} $v=u\wedge L=\phi$ \textbf{then}\label{line:indexing:f:v=u} $\,$ \textbf{return true}			
			\STATE \quad \textbf{if} $\nexists (v,\mathfrak{L})\in index[u]$ \textbf{then}\label{line:indexing:f:if}
			\STATE \quad \quad Add pair $(v, \{L\})$ into $index[u]$, then \textbf{return true}\label{line:indexing:f:if_return}
			
			\STATE \quad Let $(v,\mathfrak{L})$ represent a pair in  $index[u]$ \label{line:indexing:f:get_pair}
			\STATE \quad Remove each label set $L_i$ in $\mathfrak{L}$, if $L\subset L_i$
			\STATE \quad \textbf{if} $\nexists L_i\in\mathfrak{L} \wedge L_i\subseteq L$\textbf{then}
			\STATE \quad \quad Add $L$ into $\mathfrak{L}$, then \textbf{return true}
			\STATE \quad \textbf{return false} \label{line:indexing:function:function_ends}
			
			\rule[2pt]{7.84cm}{0.02em}
			
			\STATE \textbf{Function:} BFSTraverse($\mathcal{I}$)
			\STATE \quad $\mathbb{Q}\leftarrow\,$Initialize($\mathcal{I}$) \COMMENT{// Each element in $\mathbb{Q}$ is a queue} \label{line:bijection:initial_I_Q}
			\STATE \quad \textbf{while} $\mathbb{Q}$ is not empty \textbf{do}\label{line:bijection:while}
			\STATE \quad \quad Take an element $\mathbb{Q}_u$ from $\mathbb{Q}$ \label{line:bijection:take_queue}
			\STATE \quad \quad Take an element $v$ from $\mathbb{Q}_u$ \label{line:bijection:take_pair}
			\STATE \quad \quad \textbf{for} each edge $e=(v,l,w)$ incident to $v$ \textbf{do} \label{line:bijection:for_each_edge}
			\STATE \quad \quad \quad \textbf{if} $explored[w]$ = false \textbf{then} \label{line:bijection:explored_if}
			\STATE \quad \quad \quad \quad $w.A_\mathcal{F}\negthickspace\leftarrow\negthickspace u$, add $w$ into $\mathbb{Q}_u$, $explored[w]\negthickspace\leftarrow\negthickspace true$ \label{line:bijection:add_2}  \label{line:bijection:end_explored_if}
			
			\STATE \quad \quad \textbf{if} $\mathbb{Q}_u$ is not empty \textbf{then} \label{line:bijection:final_if}
			\STATE \quad \quad \quad Add $\mathbb{Q}_u$ into $\mathbb{Q}$ \label{line:bijection:end_while}
		\end{algorithmic}
	\end{algorithm}
	
%
%
	
	\textbf{Definitions.} An entry in the local index is related to one landmark $u$ in $\mathcal{I}$, and is formalized by $I\negthinspace I[u]\cup E\negthinspace I^T[u]\cup D[u]$. Before introducing the details about such entry, we present some local-index-related definitions. Assuming $G_u$$=$$\mathcal{F}(u)$ is a subgraph of $G$, the vertex set and edge set are denoted by $V_u$ and $E_u$, respectively. Firstly, a path set $P(s,t|G_u)$ is a subset of $P(s,t)$, where, $\forall p\in P(s,t|G_u)$, for each edge $e$ in $p$, $e\in E_u$. After that, we define of $M(s,t|G_u)$ as follows.
	
	\begin{definition}
		\label{definition:CMS_|G}
		A collection of label sets  $M(s,t|G_u)$ is a CMS from $s$ to $t$ in a subgraph $G_u$ of $G$, iff $M(s,t|G_u)=\{\mathbb{L}(p_i)|p_i\in P(s,t|G_u)\wedge\nexists p_j\in P(s,t|G_u)$, $i\neq j$, $such$ $that$ $\mathbb{L}(p_j)\subset \mathbb{L}(p_i)\}$.
	\end{definition}	
	
	Then, supposing $E_o=\{(v,l,w)|v\in V_u\wedge w\notin V_u\}$ is an edge set of $G$, where each edge in $E_o$ is incident to a vertex in $\mathcal{F}(u)$, and points to a vertex that is not in $\mathcal{F}(u)$, we define two sets of (vertex, collection of label sets) key-value pairs, $I\negthinspace I[u]$ and $E\negthinspace I[u]$. The former $I\negthinspace I[u]=\{(v,M(u,v|\mathcal{F}(u)))|v\in V_u\}$ and the latter $E\negthinspace I[u]=\big\{(w,\mathfrak{L})|\exists(v,l,w)\in E_o,\mathfrak{L}=\{L\cup l|L\in M(u,v|\mathcal{F}(u))\}\big\}$. 
	
	The meaning of the pairs in $I\negthinspace I[u]$ is obvious. For a pair $(w,\mathfrak{L})$ in $E\negthinspace I[u]$, the following statement exists: Given a label constraint $L$, if $\exists \mathcal{L}_w\negmedspace\in\negmedspace \mathfrak{L}$ and $\mathcal{L}_w\negmedspace\subseteq\negmedspace L$, then $u\negmedspace\stackrel{L}{\leadsto}\negmedspace w$. That can be easily proved by the definition of $E\negthinspace I[u]$. For the requirement of query efficiency, in practice, we reverse the pairs in $E\negthinspace I[u]$, and restore them as (label set, vertex set) key-value pairs in $E\negthinspace I^T[u]$, where $E\negthinspace I^T[u]$$=$$\big\{\forall(L,\mathcal{V})|\exists(w,\mathfrak{L}_w)\in E\negthinspace I[u]\wedge L\negmedspace\in\negmedspace\mathfrak{L}_w$, $\mathcal{V}\negthickspace=\negthickspace\{v|\forall(v,\mathfrak{L}_v)\negmedspace\in\negmedspace E\negthinspace I[u]\wedge L\negmedspace\in\negmedspace\mathfrak{L}_v\}\big\}$. Here, we further discuss $E\negthinspace I^T[u]$ in Theorem~\ref{theorem:EIT}.
	
	\begin{theorem}
		\label{theorem:EIT}
		If $L_u\subseteq L$, then $\forall v\in \mathcal{V}_u$, $u\stackrel{L}{\leadsto}v$, where $(L_u,\mathcal{V}_u)$ is a pair in $E\negthinspace I^T[u]$, and $L$ is a label constraint.
	\end{theorem}	
	\begin{proof}
		$\forall (L_u,\mathcal{V}_u)$ in $E\negthinspace I^T[u]$, for each vertex $w$ in $\mathcal{V}_u$,  there must be an entry $(w,\mathfrak{L})$ in $E\negthinspace I[u]$ that $L_u\in \mathfrak{L}$, based on the above definitions. If $L_u\negmedspace\subseteq\negmedspace L\wedge L_u\negmedspace\in\negmedspace \mathfrak{L}$, then $u\stackrel{L}{\leadsto}w$.
	\end{proof}

	Finally, to present a preliminary estimate of the correlation degree between two subgraphs $\mathcal{F}(u)$ and $\mathcal{F}(v)$ in $\mathcal{G}$, we develop $D[u]$$=$$\{(v,D(u,v)|v\in\mathcal{I})\}$, where $D(u,v)$ is the number of the pairs $(w,\mathfrak{L})$ in $E\negthinspace I[u]$ whose first element $w$ is a vertex in $\mathcal{F}(v)$.
	
	\subsubsection{Algorithm Description} \label{section:libaq:indexing:description}
	
	In this part, we introduce Algorithm~\ref{algorithm:Indexing} for the construction process of both bijection $\mathcal{F}$ and the local index entry $I\negthinspace I[u]\cup E\negthinspace I^T[u]\cup D[u]$.
	
	The landmark selection process in INS is more complicated than that in \cite{Valstar17Landmark}, where \cite{Valstar17Landmark} directly chooses the vertices with the highest degrees. According to the description of KGs in Section~\ref{section:background}, the entities with many real-world instances often relate to the vertices with relatively high degrees in a KG, as shown in Figure~\ref{fig:RDFS}. Plus, the edges that are incident to or point to such vertices often labeled by a set of RDF vocabularies, e.g. ``rdf:type'', ``rdfs:subclassOf''. Selecting the vertices with the highest degrees will cause the labels of the associated edges between the landmarks and non-landmarks to be simple. In other words, when an input label constraint do not contain any of the RDF vocabularies, the local index based on such landmarks is useless.

	In order to select a sound set of landmarks~($\mathcal{I}$), the landmark selection process~(Line~\ref{line:bijection:initial_I}) of INS is based on the RDF schema $L_S$ of the input KG $G$.  INS first randomly select a set of classes in $L_S$, then it evenly marks $k$ instances of the selected classes as landmarks~($\mathcal{I}$). Since, the smaller $|\mathcal{I}|$ is, the smaller the chances of encountering the landmarks in a query processing are, in INS, we set $|\mathcal{I}|=k=\log{|V|}\times\negmedspace\sqrt{|V|}$.

	After that, to construct $\mathcal{F}$~(Line~\ref{line:indexing:bijection}), we start a  BFS~(Function BFSTraverse, Lines~\ref{line:bijection:initial_I_Q}-\ref{line:bijection:end_while}) from all vertices in $\mathcal{I}$, simultaneously, and traverse their surrounding vertices, with following rule: For an edge $e=(v,l,w)\in E$, if both $v$ and $w$ belong to the subgraph $\mathcal{F}(u)$, $e$ is an edge of $\mathcal{F}(u)$. With the consideration of the practical requirement for efficiently searching the subgraph that a vertex $w$ belongs to, we introduce an additional attribute $A_\mathcal{F}$ for each $w$ vertex in $G$, where $w.A_\mathcal{F}\negmedspace=\negmedspace u$ represents the vertex $w$ belongs to a subgraph $\mathcal{F}(u)$~(Line~\ref{line:bijection:add_2} and Line~\ref{line:indexing:sameLPFunction}).


	Besides, function Insert$(v, L,index[u])$~(Lines~\ref{line:indexing:function:function_starts}-\ref{line:indexing:function:function_ends}) is proposed to add a pair $(v, L)$ into $index[u]$ that represents $I\negthinspace I[u]$ or $E\negthinspace I[u]$, from Line~\ref{line:indexing:function:function_starts} to Line~\ref{line:indexing:function:function_ends}. If $\nexists (v,\mathfrak{L})\in index[u]$, we add a new pair $(v,\{L\})$ into $index[u]$~(Lines~\ref{line:indexing:f:if}-\ref{line:indexing:f:if_return});  otherwise, we use $(v, L)$ to update the value of pair $(v,\mathfrak{L})$, which is demonstrated in Lines~\ref{line:indexing:f:get_pair}-\ref{line:indexing:function:function_ends}.
	
	Then, for each vertex $u$ in $\mathcal{I}$, we index $u$ in function LocalFullIndex($u$,$\mathcal{F}$)~(Lines~\ref{line:indexing:for_s}-\ref{line:indexing:function_new_ends}), in which $I\negthinspace I[u]$, $E\negthinspace I[u]$ and $E\negthinspace I^T[u]$ are empty (key, value) sets~(Line~\ref{line:indexing:empty_sets}). This function utilizes a queue in which each element is a (vertex, label set) pair. Loop~(Lines~\ref{line:indexing:loop_starts}-\ref{line:indexing:add_exteral}) starts, after adding the first pair $(u,\{\})$ into $\mathbb{Q}$, and runs until $\mathbb{Q}$ is empty. Each iteration of this loop takes a pair $(v,L)$ from $\mathbb{Q}$~(Line~\ref{line:indexing:take_out_a_pair}). We explore the edges that are incident to $v$~(Lines~\ref{line:indexing:explore_s}-\ref{line:indexing:add_exteral}), except in the following case: $(v,L)$ cannot be added into $I\negthinspace I[u]$ (Lines~\ref{line:indexing:case1_s}) by function Insert. Then, in Lines~\ref{line:indexing:explore_s}-\ref{line:indexing:add_exteral}, for each edge $(v,l,w)$, we add pair ($w$,$L\cup\{l\}$) into either $\mathbb{Q}$ or $E\negthinspace I[u]$, which is dependent on whether vertex $w$ is in $\mathcal{F}(u)$.
	
	Finally, we re-compute the elements in $E\negthinspace I[u]$ to obtain $E\negthinspace I^T[u]$ and $D[u]$ in Line~\ref{line:indexing:function_new_ends}, according to their definitions.
	
		
	\textbf{Analysis.} We study the correctness~(Theorem~\ref{theorem:index_correctness}) and complexity~(Theorem~\ref{theorem:indexing_time_complexity} and Theorem~\ref{theorem:indexing_space_complexity}) of Algorithm~\ref{algorithm:Indexing} on a KG $G=(V, E, \mathcal{L},  L_{S})$.

		\begin{theorem}
			\label{theorem:index_correctness}
			The consistency exists between the definition of an entry $I\negthinspace I[u]\cup E\negthinspace I^T[u]\cup D[u]$ and the indexing result of LocalFullIndex($u,\mathcal{F}$).
		\end{theorem}
		\begin{proof}
			$I\negthinspace I[u]$: For an element $\big(v,\mathfrak{L}\negthickspace=\negthickspace M(u,v|\mathcal{F}(u))\big)$ in $I\negthinspace I[u]$, given substructure constraint $L$, the statement is correct, iff, $\exists L'\negthickspace\in\negthickspace\mathfrak{L}\wedge L'\negthickspace\subseteq\negmedspace L\Leftrightarrow$ $u\negthickspace\stackrel{L}{\leadsto}\negthickspace v$, in $\mathcal{F}(u)$. ($\Rightarrow$) Based on function Insert, this is correct certainly. ($\Leftarrow$) If, in $\mathcal{F}(u)$, $u\negmedspace\stackrel{L}{\leadsto}\negmedspace v$, then a path $p$ exists in $\mathcal{F}(u)$ that $\mathbb{L}(p)\negmedspace\subseteq\negmedspace L$. If $\nexists L'\negthickspace\in\negmedspace\mathfrak{L}$ and $L'\negthickspace\subseteq\negmedspace L$, then at least one passed edge is not processed by function LocalFullIndex. However, in this function, we only skip the vertices in the case that all the possible path label sets from $u$ to that vertices have been searched. Thus, the assumption does not hold.
			
			$E\negthinspace I^T[u]\cup D[u]$: Based on function Insert, the consistency also exists for $E\negthinspace I[u]$. Thus, the statements about $E\negthinspace I^T[u]$ and $D[u]$ are correct~(Line~\ref{line:indexing:function_new_ends}). 
		\end{proof}
		
		
		\begin{theorem}
			\label{theorem:indexing_time_complexity}
			The indexing time complexity is $O\big(2^{|\mathcal{L}|}\times(|E|+|V|\log{2^{|\mathcal{L}|}})\big)$.
		\end{theorem}
		\begin{proof}
			Assuming, for a landmark $u$, $\mathcal{F}(u)$ contains $n$ vertices and $m$ edges. In function LocalFullIndex($u$), we push at most $n2^{|\mathcal{L}|}$ entries into the queue in this function. Each push costs $O(1)$ time and each call of function Insert takes $O(\log{2^{|\mathcal{L}|}})$ time. Thus, the indexing time complexity is $\sum\limits_{u\in\mathcal{I}}O\big((n+m)2^{|\mathcal{L}|}+n2^{|\mathcal{L}|}\log{2^{|\mathcal{L}|}}\big) 
			\leq O\big(2^{|\mathcal{L}|}\times(|E| + |V|\log{2^{|\mathcal{L}|}})\big)$. 
		\end{proof}
		
		\begin{theorem}
			\label{theorem:indexing_space_complexity}
			The indexing space complexity is $O\big(|V|2^{|\mathcal{L}|}\times(\log{|V|}+|\mathcal{L}|)\big)$.
		\end{theorem}	
		\begin{proof}
			Each entry $I\negthinspace I[u]\cup E\negthinspace I^T[u]\cup D[u]$ of the local index contains  $O(n2^{|\mathcal{L}|})$ elements, assuming $\mathcal{F}(u)$ contains $n$ vertices. The size of each element is $O(\log{|V|}+|\mathcal{L}|)$. Hence, the total indexing space complexity is $\sum\limits_{u\in\mathcal{I}}O(n2^{|\mathcal{L}|}(\log{|V|}+|\mathcal{L}|)=O\big(|V|2^{|\mathcal{L}|}\times(\log{|V|}+|\mathcal{L}|)\big)$.
			
		\end{proof}
		
		Comparing to the reported indexing time and space complexities of the traditional landmark indexing strategy~\cite{Valstar17Landmark} in Section~\ref{section:onlineSearch:LCR}, our local index could bound the indexing consumption, which is independent of the number of the chosen landmarks~($|\mathcal{I}|$).
		
	\subsection{Search Algorithm}\label{section:libaq:search}
	In order to provide INS~(Algorithm~\ref{algorithm:IML}) the ability of perceiving a better search direction, and breaking the fixed search directions of traditional uninformed search methods, we implement two data structures: a priority heap $\mathbb{H}$~(Line~\ref{line:libqa:initial_H}) and a priority queue $\mathbb{Q}$~(Line~\ref{line:libqa:initla_Q}) to be an ``evaluation function'' in our informed search strategy. Before introducing the priorities of $\mathbb{H}$ and $\mathbb{Q}$, we let $\rho(s,t)$ denote an estimate distance from a vertex $s$ to a vertex $t$, where $\rho(s,t)=D(s.A_\mathcal{F},t.A_\mathcal{F})$~(Section~\ref{section:libaq:indexing_strategy}).  

	Firstly, the main function of INS is similar to that of UIS$^*$ from Line~\ref{line:libqa:return_q_t_while_if} to Line~\ref{line:libqa:return_q_t_while_elsif}, while, in each iteration of loop~(Line~\ref{line:libqa:while_h}-\ref{line:libqa:return_q_t_while_elsif}), INS takes the top element $v$ out of $\mathbb{H}$. Note that, for two vertices $u,v\in V(S,G)$, $u$ is on the top of $v$ in $\mathbb{H}$, when: (\romannumeral1) $close[u]=F$ and $close[v]=N$; 
	(\romannumeral2) if $close[u]=close[v]=F$, $\rho(u,t)\leq\rho(v,t)$ or $u\in\mathcal{I}\wedge v\notin\mathcal{I}$; (\romannumeral3) if $close[u]=close[v]=N$, $\rho(s,u)\leq\rho(s,v)$ or $u\in\mathcal{I}\wedge v\notin\mathcal{I}$;


	Similar to UIS$^*$, INS also initialize a global variable throughout the algorithm. Instead of implementing a stack, we use a queue $\mathbb{Q}$ to increase the probability of encountering the vertices in $\mathcal{I}$ of the bijection $\mathcal{F}$, as the local index covers a portion of the full TC~(Section~\ref{section:onlineSearch:LCR}) of the input KG $G$. Even though $\mathbb{Q}$ is initialized in the main function with an element $s$~(Line~\ref{line:libqa:initla_Q}), the other operations of adding an element into $\mathbb{Q}$ or taking an element out of $\mathbb{Q}$ are in function LCS($s^*,t^*,L,B$). Besides, for two vertices $u$ and $v$ in $\mathbb{Q}$, it stores $u$ in front of $v$, when: 
		(\romannumeral1) $close[u]\negthickspace=\negthickspace T\wedge close[v]\negthickspace=\negthickspace F$;
		(\romannumeral2) $u.A_\mathcal{F}\negmedspace=\negmedspace t^*\negthickspace.A_\mathcal{F}\negmedspace\neq\negmedspace v.A_\mathcal{F}$; 
		(\romannumeral3) $u\negmedspace\in\negmedspace\mathcal{I}\wedge v\negmedspace\notin\negmedspace\mathcal{I}$; 
		(\romannumeral4) $\rho(u,t^*)\negmedspace\leq\negmedspace\rho(v,t^*)$; 
		(\romannumeral5) $u,v\negmedspace\notin\negmedspace\mathcal{I}$, $close[l_u]\negthickspace=\negthickspace N\wedge close[l_v]\negthickspace\neq\negthickspace N$;
		(\romannumeral6) otherwise, $u$ is added into $\mathbb{Q}$ before $v$. 
	
	Note that, for two elements $x$ and $y$ in $\mathbb{Q}$, if $x$ and $y$ represent a same vertex in $G$, $\mathbb{Q}$ deletes the first added element. Plus, according to the above $\mathbb{Q}$ priority rules,  if $close[x]\negthickspace=\negthickspace T\wedge close[y]\negthickspace=\negthickspace F$, $x$ is in the front of $y$ in $\mathbb{Q}$. Then, in INS, the elements, which are added into $\mathbb{Q}$ in the k$^{th}$ function LCS invocation with parameter $B$$=$$T$, are processed first in this invocation.  
		
	In function LCS($s^*,t^*,L,B$), for each element $u$ taken out from $\mathbb{Q}$~(Line~\ref{line:libqa:lcs:while_take}), we traverse the edges~($e$$=$$(u,l,w)\wedge l\negthickspace\in\negthickspace L$) that are incident to $u$ by a loop~(Lines~\ref{line:libqa:lcs:while_for}-\ref{line:libqa:lcs:for_if3_return}). Initially, if $t^*\negthickspace.A_\mathcal{F}\negmedspace=\negmedspace w$~(subgraph $\mathcal{F}(u)$ contains vertex $t^*$), in Line~\ref{line:libqa:lcs:for_if}, we implement a function Check($I\negthinspace I[w]$,$t^*$) to evaluate the existence of $w\stackrel{L}{\leadsto}t^*$, where if $\exists (t^*,\mathfrak{L})\negthickspace\in\negthickspace I\negthinspace I[w]\wedge \exists L_i\negthickspace\in\negthickspace \mathfrak{L}$ and $L_i\negthickspace\subseteq\negthickspace L$, function Check returns true, otherwise, false.
	
	\begin{algorithm}[t]
		\caption{INS Algorithm($G$, $Q$)}
		\label{algorithm:IML}
		\begin{algorithmic}[1]
			\renewcommand{\algorithmicrequire}{\textbf{Input:}}
			\renewcommand{\algorithmicensure}{\textbf{Output:}}
			\renewcommand{\algorithmiccomment}[1]{// #1}
			
			\REQUIRE $G=(V, E, \mathcal{L}, L_{S})$ represents a KG, \\$Q=(s,t,L,S)$ is a LSCR query,\\ $V(S,G)$ is obtained by implementing a SPARQL engine.
			\ENSURE  The answer of Q.
			
			\STATE Let $\mathbb{H}$ be a priority heap initialized by $V(S,G)$\label{line:libqa:initial_H}
			\STATE Let $\mathbb{Q}$ be a global priority queue with an element $s$\label{line:libqa:initla_Q}
			
			\COMMENT{The priorities of the elements in both $\mathbb{H}$ and $\mathbb{Q}$ are depnedent on $s$, $t$, $L_S$ and $\mathcal{F}$ }
			
			\STATE $close[s]\leftarrow F$  \COMMENT{The initial values in $close$ are set to N}\label{line:libqa:initial_s_F}
			
			\STATE\textbf{while} $\mathbb{H}$ is not empty \textbf{do}\label{line:libqa:while_h}
			\STATE \quad Take the top element $v$ from $\mathbb{H}$ \label{line:libqa:while_take}
			\STATE \quad \textbf{if} $close[v]=N$ \textbf{then} \label{line:libqa:while_if}
			\STATE \quad \quad \textbf{if} $v=s$ \textbf{or} $v=t$ \textbf{then} \label{line:libqa:while_if_if}
			\STATE \quad \quad \quad \textbf{return} $Q=\,$LCS$(s,t,L,F)$ \label{line:libqa:return_t_v}
			\STATE \quad \quad \textbf{else if} LCS$(s,v,L,F)$ \textbf{then} \label{line:libqa:while_if_elsif}
			\STATE \quad \quad \quad \textbf{if} LCS$(v,t,L,T)$ \textbf{then}  \label{line:libqa:while_if_elsif_if}
			\STATE \quad \quad \quad \quad \textbf{return} $Q=T$ \label{line:libqa:return_q_t_while_if}
			\STATE \quad \textbf{else if} $close[v]=F$ \textbf{then} \label{line:libqa:while_elsif}
			\STATE \quad \quad \textbf{if} LCS$(v,t,L,T)$ \textbf{then} \label{line:libqa:while_elsif_if}
			\STATE \quad \quad \quad \textbf{return} $Q=T$ \label{line:libqa:return_q_t_while_elsif}
			
			\STATE \textbf{return} $Q=F$  \label{line:libqa:final_q_f}
			
			\rule[2pt]{7.85cm}{0.02em}
			
			\STATE	\textbf{Function:} LCS$(s^*, t^*, L, B)$  \label{line:libqa:lcs:function_start} \COMMENT{Verify whether $s^*\negthickspace\stackrel{L}{\leadsto}\negthickspace t^*$}
			\STATE \quad \textbf{if} $B=T$ \textbf{then}  \label{line:libqa:lcs:initial_if}
			\STATE \quad \quad $close[s^*]\leftarrow T$, add $s^*$ into $\mathbb{Q}$  \label{line:libqa:lcs:initial_if_add_t}
			\STATE \quad \textbf{while} $B=F\wedge \mathbb{Q}\neq \phi$ \textbf{or} $B=close[\mathbb{Q}.first]=T$ \textbf{do}  \label{line:libqa:lcs:while}
			\STATE \quad \quad Take an element $u$ from $\mathbb{Q}$  \label{line:libqa:lcs:while_take}
			\STATE \quad \quad \textbf{for} each edge $e=(u,l,w)$, $l\in L$, incident to $u$ \textbf{do}  \label{line:libqa:lcs:while_for}
			
			\STATE \quad \quad \quad \textbf{if} $t^*\negthickspace.A_\mathcal{F}\negthickspace=\negthickspace w$ \textbf{and} Check($I\negthinspace I[w], t^*$) \textbf{then}  \label{line:libqa:lcs:for_if}
			\STATE \quad \quad \quad \quad \textbf{return true}  
			
			\STATE \quad \quad \quad \textbf{if} $w\in\mathcal{I}$ \textbf{then}  \label{line:libqa:lcs:for_if2} \COMMENT{Pruning the search space}
			\STATE \quad \quad \quad \quad Cut($I\negthinspace I[w]$) \textbf{and} Push($E\negthinspace I^T[w]$) 
			
			\STATE \quad \quad \quad \textbf{else if} $close[w]$$=$$N\:$\textbf{or}$\:close[w]\negthickspace=\negthickspace F$$\wedge$$B$$=$$T$ \textbf{then}  \label{line:libqa:lcs:for_elsif}	
			\STATE \quad \quad \quad \quad Add $w$ into $\mathbb{Q}$, $close[w]\leftarrow B$ \label{line:libqa:lcs:for_elsif_add}
			
			\STATE \quad \quad \quad \textbf{if} $w=t^*$ \textbf{then}  \label{line:libqa:lcs:for_if3}
			\STATE \quad \quad \quad \quad \textbf{return true}  \label{line:libqa:lcs:for_if3_return}
			
			\STATE \quad \textbf{return false}  \label{line:libqa:lcs:final_return}
		\end{algorithmic}
	\end{algorithm}
	
	Then, in Line~\ref{line:libqa:lcs:for_if2}, if $w\in \mathcal{I}$, two functions in LCS, Cut($I\negthinspace I[w]$) and Push($E\negthinspace I^T[w]$), are introduced with the precomputed local index $I\negthinspace I[w]$ and $E\negthinspace I^T[w]$, to prune the space and improve the query efficiency. For a pair $(x,\mathfrak{L})$ in $I\negthinspace I[w]$, the former set $close[x]$ to $B$, if $close[x]$$\neq$$T$ and $\exists L_i\negthickspace\in\negthickspace \mathfrak{L}\wedge L_i\negthickspace\subseteq\negthickspace L$; for a pair $(L_x,\mathcal{V})$ in $E\negthinspace I^T[w]$, the latter adds each vertex $u\in\mathcal{V}$ into $\mathbb{Q}$, and changes $close[u]$ to $B$, if $L_x\negthickspace\subseteq\negthickspace L\wedge B\negthickspace=\negthickspace T\wedge close[u]\negthickspace\neq\negthickspace T$ or $L_x\negthickspace\subseteq\negthickspace L\wedge B\negthickspace=\negthickspace F\wedge close[u]\negthickspace =\negthickspace N$. After that, in Line~\ref{line:libqa:lcs:for_elsif}, if $w\notin \mathcal{I}$, we add $w$ into $Q$ and set $close$ to $B$, in the cases that $close[w]$$=$$N$ or $close[w]\negthickspace=\negthickspace F \wedge B$$=$$T$, which is same to UIS$^*$. 

	\textbf{Analysis.} We conduct the analysis on a KG $G$$=$$(V, E, \mathcal{L}$, $L_{S})$, where the query time complexity and the correctness are proved in Theorem~\ref{theorem:ins_query_time_complexity} and Theorem~\ref{theorem:ins_query_correctness}, respectively.
	
	\begin{theorem}
		\label{theorem:ins_query_time_complexity}
		The time complexity of the search strategy is $O\big(|E|+|V|\times(\log{|V|}+2^{|\mathcal{L}|})\big)$.
	\end{theorem}
	\begin{proof}
		Each entry $I\negthinspace I[u]\negthinspace\cup\negthinspace E\negthinspace I^T[u]\negthinspace\cup\negthinspace D[u]$ in the local index contains $O(n2^{\mathcal{L}})$ elements, supposing $\mathcal{F}(u)$ contains $n$ vertices. In the worst case, we visit all the landmarks of $\mathcal{I}$ in one query processing so that the total local index visiting time complexity is  $O(|V|2^{\mathcal{L}})$. We at most traverse the input KG $G$ twice. Thus, the graph traversing time complexity is $O(|V|+|E|)$. Besides, we perform a complete ordering of the $|V|$ elements in $\mathbb{Q}$, but the above operation can occur once only, so its time complexity is $O(|V|\log{|V|})$. Overall, the total complexity is $O(|E|+|V|\times(\log{|V|}+2^{|\mathcal{L}|}))$.
	\end{proof}
	
	\begin{theorem}
		\label{theorem:ins_query_correctness}
		$Q$ is a true query, iff, INS returns true.
	\end{theorem}
	\begin{proof}
		INS is different to UIS$^*$ in the following aspects: (\romannumeral1) INS utilize a local index constrcuted by Algorithm~\ref{algorithm:Indexing}; (\romannumeral2) INS processes the elements in $V(S,G)$ with orders; (\romannumeral3) INS applies the global variable with a priority queue instead of a stack, and INS does not contain the operation that is in Line~\ref{line:sm:lcs:function_end} of UIS$^*$. According to Theorem~\ref{theorem:index_correctness}, both (\romannumeral1) and (\romannumeral2) do not affect the correctness of INS. Then, the global variable could automatically remove duplicate elements, and keep the elements with the values in the surjection $close$ are $T$ in the front of the queue. Hence, INS still holds the correctness.
	\end{proof}

\section{Experimental Results} \label{section:result}
	In this section, we evaluate the performances of the proposed algorithms: UIS, UIS$^*$ and INS, on both synthetic and real KGs. Specially, since the effect of the label constraints on the query performance has been widely studied~\cite{Jin2010CLR,Valstar17Landmark}, in this section, we are interested in: the indexing time and space consumption of our local index; the influence of the KG scales on the query performance; the impact of the substructure constraints on the query efficiency, including the selectivity of the substructure constraint and the number of the vertices that could satisfy the substructure constraint.  
	
	After introducing the settings and measures in experimental evaluation as follows, we detailedly illustrate the experiments on synthetic and real datasets in Section~\ref{section:r:lubm} and Section~\ref{section:r:yago}, respectively. In addition, based on all the experimental results, we present a discussion in Section~\ref{section:r:discussion}.
		
	\textbf{Settings.} Our experiments run on a machine with the intel i7-9700 CPU, 64 GB RAM, and 1TB disk space. We implement all the proposed algorithms in this paper with Java. Plus, to avoid the effect of the implementation language, the traditional landmark indexing method~\cite{Valstar17Landmark} is also implemented with Java, and the indexing parameters $k$ and $b$ of \cite{Valstar17Landmark} are set to $1250\negmedspace+\negmedspace\sqrt{|V|}$ and $20$, respectively, as that in the experiments of \cite{Valstar17Landmark}. Note that, the indexes that build by local index and \cite{Valstar17Landmark} are stored by the same data structure and on disk. 
		
	Then, we utilize the SPARQL engine that is present in \cite{LKAQ} for both UIS$^*$ and INS, which is also rewritten with Java. The adoption of the SPARQL engine is due to its high efficiency in the experimental setting of this paper. As that provides the query results with the controlling of parameter $UNI_{Max}$, $Max$ and $E\delta$, in order to obtain the full set of $V(S,G)$, according to \cite{LKAQ}, we set the parameters $UNI_{Max}$ and $Max$ to $+\infty$, and set $E\delta$ to $1$. 
		
	Besides, the datasets used in this experiments are the Lehigh University Benchmark~(LUBM)~\cite{LUBM} and YAGO~\cite{YAGO}. The former is a synthetic benchmark based on an ontology of the university domain, and provides a scalable synthetic data generator; the latter is a huge semantic KG, which is derived from Wikipedia, WordNet, etc., and contains about 10 million entities, and 120 million facts. 
		
	\textbf{Query performance measures.} For a group of the evaluation queries in the following experiments, the measures of the query performance include: (\romannumeral1) the average running time of the queries; (\romannumeral2) the average number of the vertices whose states in $close$ are not $N$. The reason for the former measure is obvious, while, we utilize the latter, instead of the number of the nodes in the search tree $\mathbb{T}$~($|\mathbb{T}|$), is because $|\mathbb{T}|$ can be directly reflected by the former measure.   
		
	\begin{table}[t]
		\caption{Synthetic datasets.}
		\label{tab:vertex_edge_number}
		\begin{center}
			\setlength{\tabcolsep}{1.5mm}{
				\begin{tabular}{c|cc|cc|cc}
					\hline \multirow{2}*{Dataset}& \multirow{2}*{Vertex} &\multirow{2}*{Edge} & \multicolumn{2}{c|}{Local index}& \multicolumn{2}{c}{Traditional}\\
					
					&	& & IT(s) & IS(MB)& IT(s) & IS(MB)\\
					\hline
					D0	&   $0.06$M		&	$0.23$M		& 	23		&	4	& 27,171&	11,700\\
					D1  & 	$3.7$M 		&	$13.3$M 	&	1,540	&	136	&-	&-\\
					D2  & 	$7.5$M 		&	$26.6$M 	&	3,119	&	273	&-	&-\\
					D3  &	$11.3$M 	& 	$40.0$M 	&	4,538	&	411	&-	&-\\
					D4  & 	$15.1$M 	& 	$53.4$M 	&	6,203	&	546	&-	&-\\
					D5  &	$18.9$M 	& 	$66.7$M 	&	7,699	&	684	&-	&-\\
					\hline				
				\end{tabular}
			}
		\end{center}
	\end{table}	
	
	\begin{table}[t]
		\caption{Five typical subgraph constraints on the synthetic datasets.}
		\label{tab:SPARQL_query}
		\begin{center}
			\renewcommand\arraystretch{1.4}
			\scalebox{0.93}[0.93]{
				\begin{tabular}{cp{7.9cm}}
					\hline
					\textbf{Name}&\textbf{Description in SPARQL} \\
					\hline
					S1& SELECT ?x WHERE \{ ?x $\langle$ub:researchInterest$\rangle$ `Research12'.\}\\
					
					S2& SELECT ?x WHERE \{ ?x $\langle$ub:researchInterest$\rangle$ `Research12'. ?x $\langle$rdf:type$\rangle$ $\langle$ub:AssociateProfessor$\rangle$. \}\\    			
					
					S3& SELECT ?x WHERE \{?x $\langle$rdf:type$\rangle$ $\langle$ub:UndergraduateStudent$\rangle$. ?x $\langle$ub:takesCourse$\rangle$ ?y. ?y $\langle$rdf:type$\rangle$ $\langle$ub:Course$\rangle$. \}\\
					
					S4& SELECT ?x WHERE \{?x $\langle$ub:name$\rangle$ `GraduateStudent4'. ?x $\langle$ub:takesCourse$\rangle$ ?y$_1$. ?x $\langle$ub:advisor$\rangle$ ?y$_2$. ?x $\langle$ub:memberOf$\rangle$ ?y$_3$. ?z$_1$ $\langle$ub:takesCourse$\rangle$ ?y$_1$. ?y$_2$ $\langle$ub:teacherOf$\rangle$ ?z$_2$. ?y$_2$ $\langle$ub:worksFor$\rangle$ ?z$_3$. ?y$_3$ $\langle$ub:subOrganizationOf$\rangle$ ?z$_4$. \}\\
					
					S5& SELECT ?x WHERE \{?x $\langle$ub:emailAddress$\rangle$ `FullProfessor0@Department0.University0.edu'. ?x $\langle$ub:undergraduate-DegreeFrom$\rangle$ ?y1. ?x $\langle$ub:mastersDegreeFrom$\rangle$ ?y2. ?x $\langle$ub:doctoralDegreeFrom$\rangle$ ?y3.\}\\
					\hline
				\end{tabular}
			}
		\end{center}
	\end{table}
				
	\subsection{Experiments on LUBM}\label{section:r:lubm}
	
	We randomly generate datasets D0-D5~(Table~\ref{tab:vertex_edge_number}), by applying LUBM. Even though D0 contains $40$k vertices and $230$k edges, and is only utilized to compare the indexing consumption of our local index and the traditional landmark indexing~\cite{Valstar17Landmark}, its scale still greatly exceeds that of the most datasets used in~\cite{Jin2010CLR,Valstar17Landmark}. Based on D0-D5, we evaluate the indexing time and space consumption of both local index and the traditional landmark indexing method~\cite{Valstar17Landmark}. The results are demonstrated in Table~\ref{tab:vertex_edge_number}, which confirm our theoretical analysis that, compared to \cite{Valstar17Landmark}, the indexing consumption of our local index is affordable, and increases linearly with the input KG scale. Note that, the indexing processes are limited within eight hours.
			
	After that, we conduct groups of experiments, and each group is related to a certain substructure constraint in S1-S5~(expressed in SPARQL, Table~\ref{tab:SPARQL_query}), and runs on D1-D5. The reason, why we do not consider the queries provided in LUBM~\cite{LUBM}, is that they are designed for evaluating the performance of the RDF engine query optimizers and the \textit{not-ontology-reasoning}~\cite{LinkDataBook}, and are useless for our evaluation. 
	
	We demonstrate the query generation method and the experimental result analysis for the above groups in Section~\ref{section:r:lubm:query_generation} and Section~\ref{section:r:lubm:analysis}, respectively, after introducing the characteristics of S1-S5. Assuming $D$ is a dataset generated by LUBM, whose vertex set is $V$. S1 is a simple and matches only the vertices whose research interest is `Research12'. We treat S1 as our substructure constraint baseline among S1-S5, where  $\frac{|V(S1,D)|}{|V|}\negthickspace\approx\negthickspace 1\tcperthousand$. Comparing to S1, S2 contains a normal selectivity, which matches the vertices which also should be associate professors, and $\frac{|V(S2,D)|}{|V(S1,D)|}\negthickspace\approx\negthickspace 50\%$. Then, we devise S3, where $\frac{|V(S3,D)|}{|V(S1,D)|}\negthickspace\approx\negthickspace 120$, which means there is a large set of vertices in $D$ that could satisfy S3. Then, S4 is related to high selectivity and $\frac{|V(S4,D)|}{|V(S1,D)|}\negthickspace\approx\negthickspace 1$, where the vertices in $|V($S4$,D)|$ are named `GraduateStudent4', take a course $?y_1$, etc. Specially, $|V($S5$,D)|\negthickspace=\negthickspace1$.

	\subsubsection{Evaluation query generation} \label{section:r:lubm:query_generation}
	Generating an evaluation query $Q\negthickspace=\negthickspace(s,t,L,S)$ in our experiments is intricate, because, if $s$ reaches $t$ only with a few steps, $Q$ cannot test the limits of the search algorithms. Plus, due to the existence of the irrelevant factors that affect query performance, we must ensure our experiments are immune to such irrelevant factors.
	

	\textbf{Method of controlling irrelevant variable.} Apparently, label constraint is one of the main factors that affect the query efficiency, but we do not focus on that in the experiments, because the impact of the label constraints on the query efficiency is widely studied in the LCR works. Supposing $\mathfrak{t}$ is the size of an input KG label set, for each group of queries, we guarantee the sizes of the label constraints are in the range of $[0.2\mathfrak{t}, 0.8\mathfrak{t}]$, and ensure a uniform distribution of the such sizes among the ranges of $[0.2\mathfrak{t}, 0.4\mathfrak{t})$, $[0.4\mathfrak{t}, 0.6\mathfrak{t})$ and $[0.6\mathfrak{t}, 0.8\mathfrak{t}]$. Besides, for a group of false-LSCR queries, the performances of the search methods may be affected by the types of false queries. For example, a false-LSCR query $Q$~($s\negthickspace\stackrel{L,S}{\nrightarrow}\negthickspace t$) has three possibilities: $s\negthickspace\stackrel{L}{\nrightarrow}\negthickspace t\negmedspace\wedge\negmedspace s\negthickspace\stackrel{S}{\leadsto}\negthickspace t$, $s\negthickspace\stackrel{L}{\leadsto}\negthickspace t\negmedspace\wedge\negmedspace s\negthickspace\stackrel{S}{\nrightarrow}\negthickspace t$ and $s\negthickspace\stackrel{L}{\nrightarrow}\negthickspace t\negmedspace\wedge\negmedspace s\negthickspace\stackrel{S}{\nrightarrow}\negthickspace t$. Thus, for a group of the generated false queries, we also guarantee a uniform proportion of the above possibilities.

	\textbf{Query generation method.}  For a substructure constraint Si of Table~\ref{tab:SPARQL_query} and a dataset Dj of Table~\ref{tab:vertex_edge_number}, we generate two groups of LSCR queries: $1,000$ true-queries~($\mathcal{Q}_t$) and $1,000$ false-queries~($\mathcal{Q}_f$). Assuming $Q\negthickspace=\negthickspace(s,t,L,$Si$)$, the process starts with randomly selecting $s$ and $L$ on Dj. Due to the restriction for controlling the irrelevant variables, if $Q$ cannot be added into $\mathcal{Q}_t$ or $\mathcal{Q}_f$, we re-generate $L$. Then, to select a sound target vertex $t$ of $Q$, we start a BFS from $s$, and stop it after $\log|V|$ iterations, after which $t$ is a BFS-unexplored vertex. This is for filtering out the vertices that $s$ reaches only with a few steps. In addition, we apply UIS to find out whether $Q$ is a true query~($Q\negthickspace=\negthickspace T$) or not~($Q\negthickspace=\negthickspace F$), and record the number of the nodes in the search tree, denoted by $|\mathbb{T}|$. With a random number $min$ in $[10\log{|V|}, \frac{|V|}{10\log{|V|}}]$, if $|\mathbb{T}|<min$, we discard $Q$. If $Q\negthickspace=\negthickspace T$ and $|\mathcal{Q}_t|\negthickspace<\negthickspace1,000$, we add $Q$ into $\mathcal{Q}_t$; if $Q\negthickspace=\negthickspace F\wedge|\mathcal{Q}_f|\negthickspace<\negthickspace 1,000$ and $\mathcal{Q}_f\negmedspace\cup\negmedspace \mathcal{Q}$ satisfies the above restriction, $Q$ belongs to $\mathcal{Q}_f$.

	\begin{figure}[t]
		\centering
		\renewcommand{\thesubfigure}{}
		\subfigure[(a) True]{
			\includegraphics[scale = 0.6]{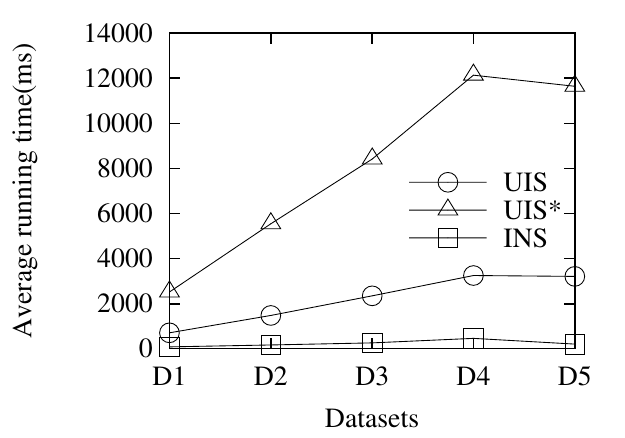}}
		\subfigure[(b) False]{
			\includegraphics[scale = 0.6]{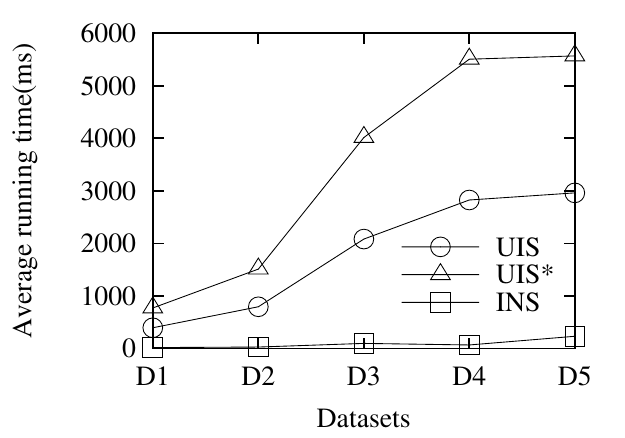}}
		\subfigure[(c) True]{
			\includegraphics[scale = 0.6]{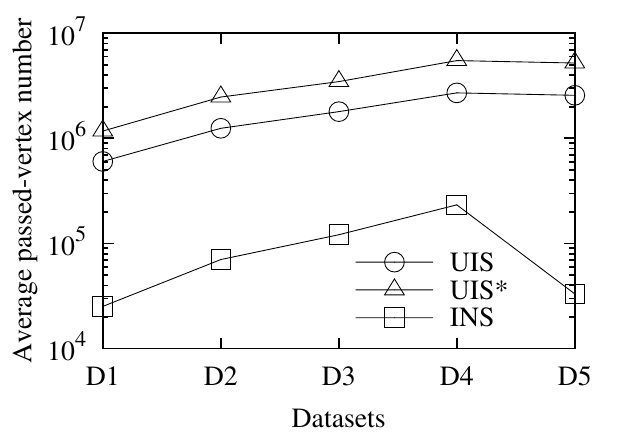}}
		\subfigure[(d) False]{
			\includegraphics[scale = 0.6]{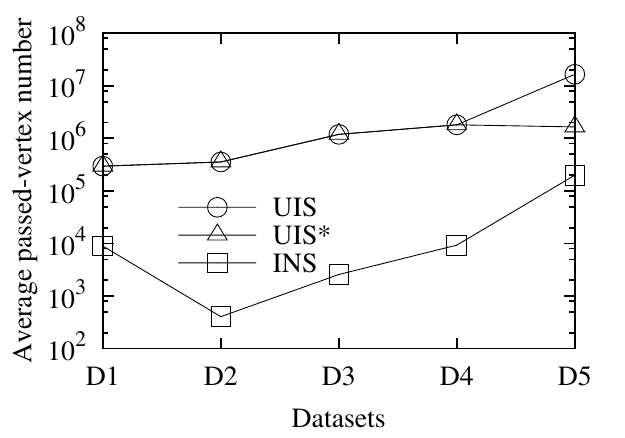}}
		\caption{Results related to the substructure constraint S1.}
		\label{fig:lubm_s1}
	\end{figure}

	\subsubsection{Experimental result analysis}\label{section:r:lubm:analysis}
	
	\begin{figure}[t]
		\centering
		\renewcommand{\thesubfigure}{}
		\subfigure[(a) True]{
			\includegraphics[scale = 0.6]{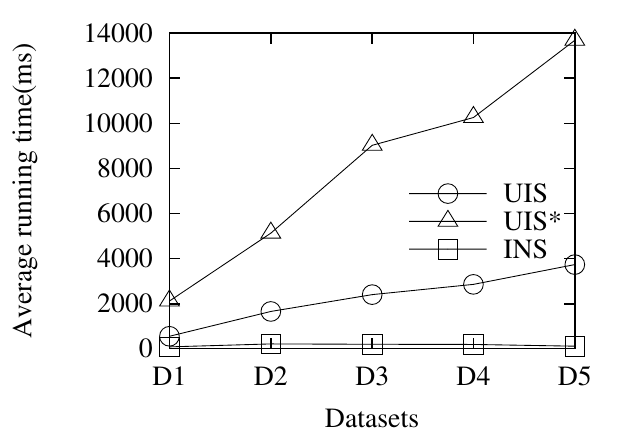}}
		\subfigure[(b) False]{
			\includegraphics[scale = 0.6]{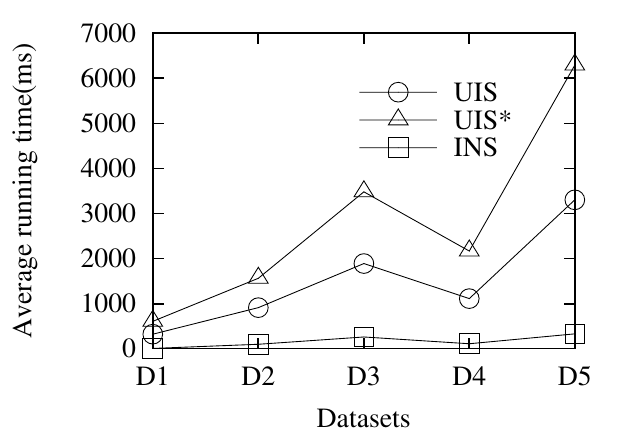}}
		\subfigure[(c) True]{
			\includegraphics[scale = 0.6]{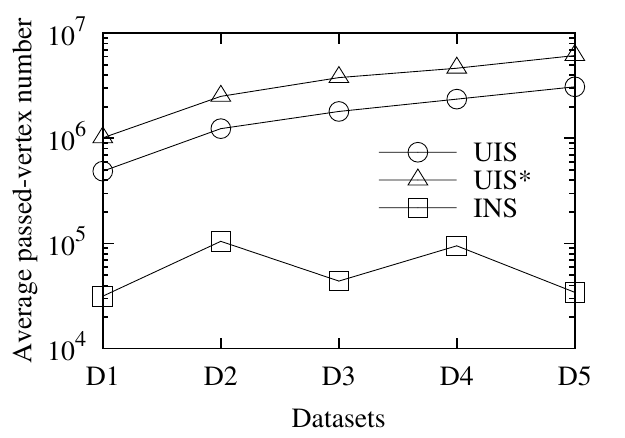}}
		\subfigure[(d) False]{
			\includegraphics[scale = 0.6]{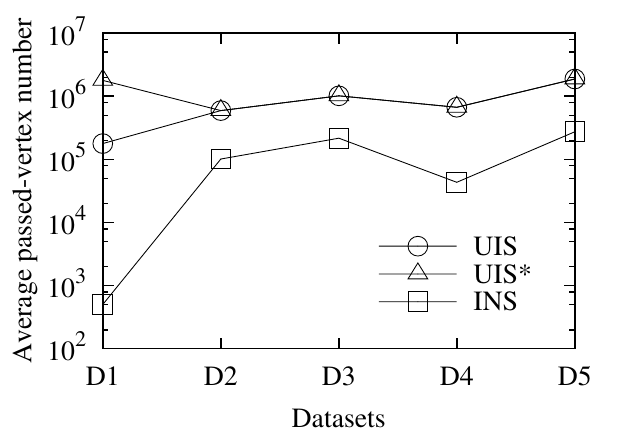}}
		\caption{Results related to the substructure constraint S2.}
		\label{fig:lubm_s2}
	\end{figure}
	
	For the results about UIS on LSCR queries related to S1, the average running times, of both true~(Figure~\ref{fig:lubm_s1}(a)) and false~(Figure~\ref{fig:lubm_s1}(b)) queries, all show a linear upward trend with the increase of the input KG scale. The average passed-vertex number of UIS on such queries is about $10^6$~(Figure~\ref{fig:lubm_s1}(c) and Figure~\ref{fig:lubm_s1}(d)). Comparing to S1, the query performances of UIS, corresponding to S2~(normal selectivity) and S4~(high selectivity), are basically unchanged, according to Figure~\ref{fig:lubm_s2} and Figure~\ref{fig:lubm_s4}. However, for the true-LSCR queries under the substructure constraint S3~(Figure~\ref{fig:lubm_s3}), UIS consumes five times running time, compared to S1. Plus, the ranges of the average running times, of both true and false queries corresponding to S5~(Figure~\ref{fig:lubm_s5}), are $[5$s$,35$s$]$ and $[1$s$,9$s$]$, respectively, while the average passed-vertex number is still around $10^6$ of either true or false queries. That is because, the bigger $|V($Si$,$Dj$)|$ is, the earlier UIS meets a vertex $v$~($v\in V(S,G)$) in a search path. Such results reflect that $|V($Si$,$Dj$)|$ effects the performance of UIS to some extent.
	
	\begin{figure}[t]
		\centering
		\renewcommand{\thesubfigure}{}
		\subfigure[(a) True]{
			\includegraphics[scale = 0.6]{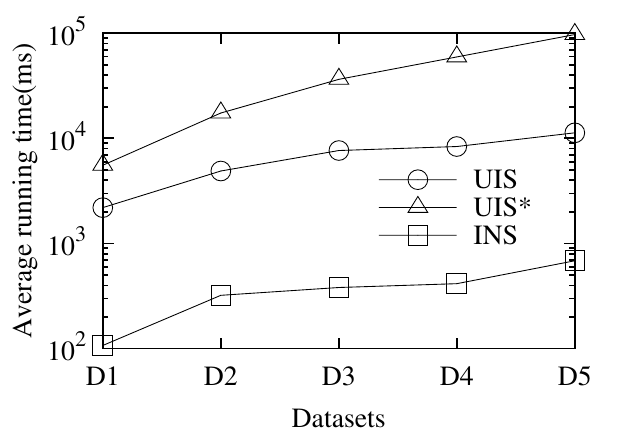}}
		\subfigure[(b) False]{
			\includegraphics[scale = 0.6]{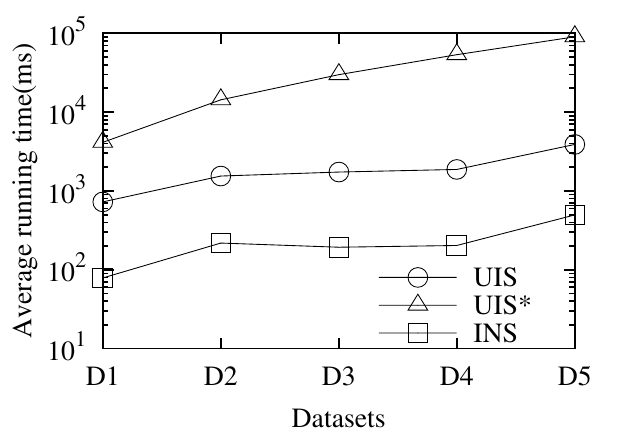}}
		\subfigure[(c) True]{
			\includegraphics[scale = 0.6]{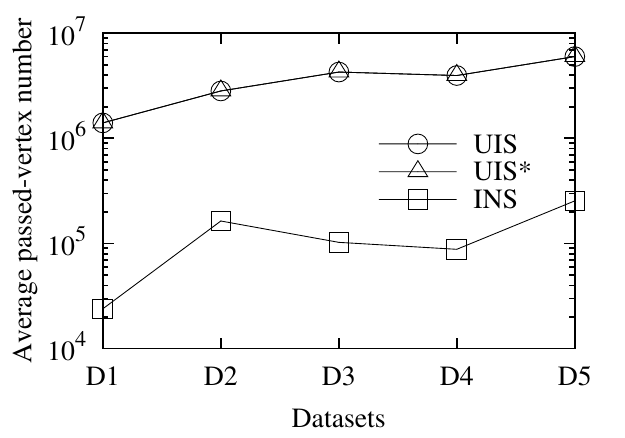}}
		\subfigure[(d) False]{
			\includegraphics[scale = 0.6]{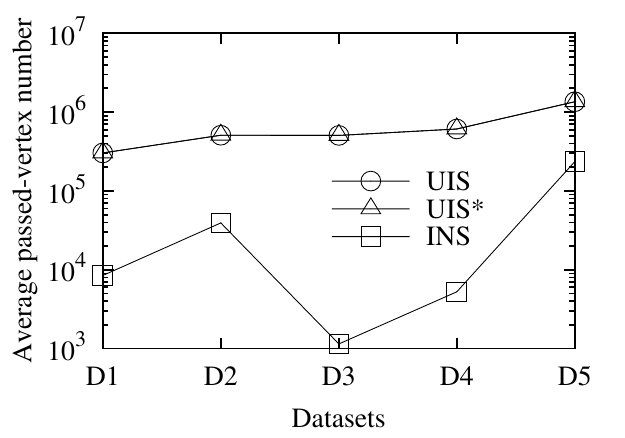}}
		\caption{Results related to the substructure constraint S3.}
		\label{fig:lubm_s3}
	\end{figure}
	
		\begin{figure}[t]
		\centering
		\renewcommand{\thesubfigure}{}
		\subfigure[(a) True]{
			\includegraphics[scale = 0.6]{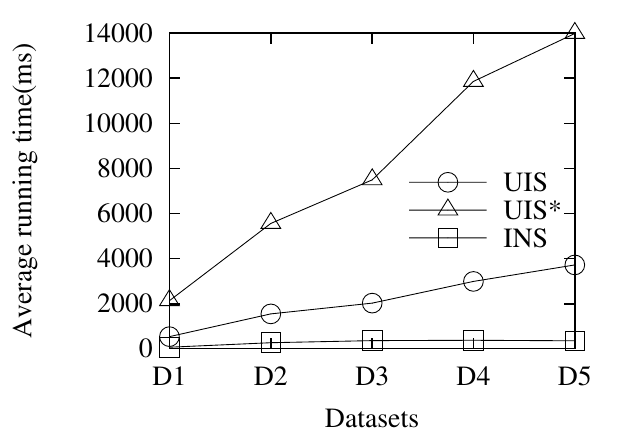}}
		\subfigure[(b) False]{
			\includegraphics[scale = 0.6]{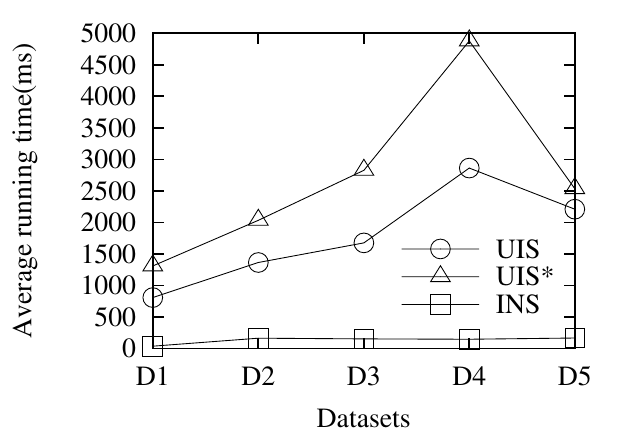}}
		\subfigure[(c) True]{
			\includegraphics[scale = 0.6]{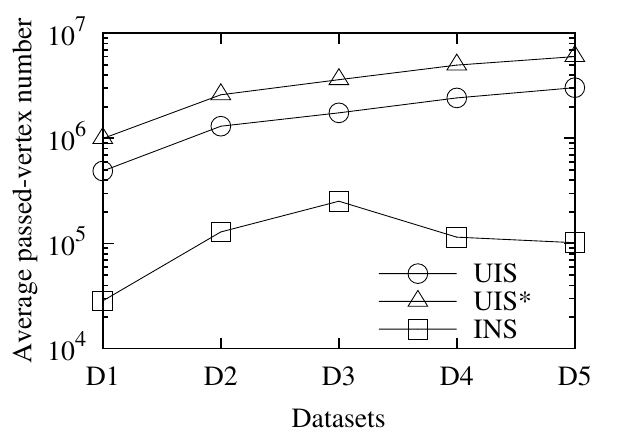}}
		\subfigure[(d) False]{
			\includegraphics[scale = 0.6]{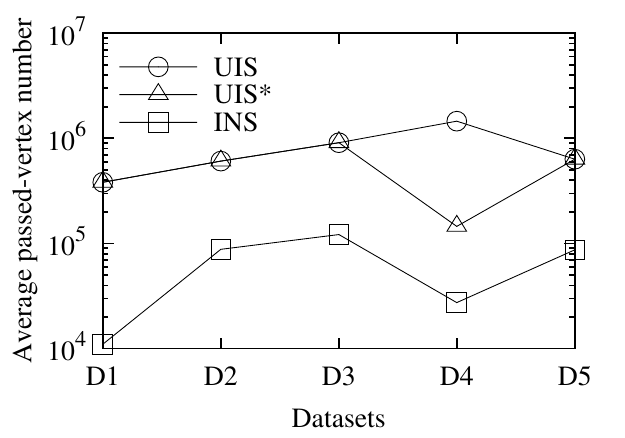}}
		\caption{Results related to the substructure constraint S4.}
		\label{fig:lubm_s4}
	\end{figure}
	
	We observe that the average running times of both UIS$^*$ and INS also grow linearly with the increase of the KG scale. The substructure constraints with normal~(S2, Figure~\ref{fig:lubm_s2}) and high~(S4, Figure~\ref{fig:lubm_s4}) selectivity do not show a significant effect on the query efficiency of either UIS$^*$ or INS, comparing to S1. Furthermore, based on the results shown in Figure~\ref{fig:lubm_s3} and Figure~\ref{fig:lubm_s5}, the number of the vertices in an input KG that satisfy the given substructure constraint also effects the query performance of both UIS$^*$ and INS. Then, even though UIS$^*$ is an improved algorithm of UIS, the performance of UIS$^*$ on actual queries is not satisfactory, and the linear increase slope of UIS$^*$ is greater than that of UIS, except in the case of S5~(Figure~\ref{fig:lubm_s5}). That is because the order of the vertices in a $V(S,G)$ is random, causing that UIS$^*$ always falls into a bad direction. The results shown in Figure~\ref{fig:lubm_s3} and Figure~\ref{fig:lubm_s5} demonstrate that the impact of the bad direction~(Theorem~\ref{theorem:sqlcs_weakness}) is much greater than that of the high selectivity. However, with a local index pruning the search space and selecting a better search direction, INS overcomes the above disadvantage of UIS$^*$, and achieves a much higher efficiency than UIS and UIS$^*$, based on experimental results reported in Figures~\ref{fig:lubm_s1}-\ref{fig:lubm_s5}.
	
	\subsection{Experiments on YAGO}\label{section:r:yago}
	
	The utilized dataset YAGO\footnote{We download the dataset from ``https://www.mpi-inf.mpg.de /departments/databases-and-information-systems/research/yago-naga/yago/archive/''.} includes about $4$M vertices and $13$M edges. The indexing time and space consumption of local index on YAGO are $4,993$s and $86$MB, respectively.
	
	\textbf{Evaluation query generation.} We develop sets of experiments on the real KG YAGO~(denoted by $G$). Each group of the experiments also corresponds to two groups of LSCR queries: $1,000$ true-queries~($\mathcal{Q}_t$) and $1,000$ false-queries~($\mathcal{Q}_f$). Different from that on LUBM, we randomly generate the substructure constraints of both $\mathcal{Q}_t$ and $\mathcal{Q}_f$, where the number of the vertices satisfying such constraints on $G$ is in the same order of magnitude $m$~($m\in\{10^1,10^2,\dots\}$). 
	
		\begin{figure}[t]
		\centering
		\renewcommand{\thesubfigure}{}
		\subfigure[(a) True]{
			\includegraphics[scale = 0.6]{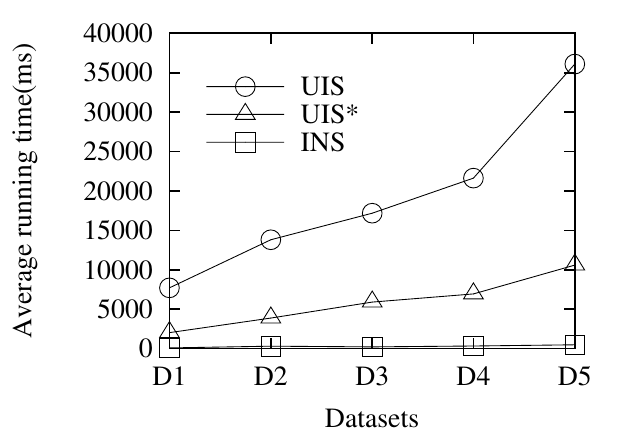}}
		\subfigure[(b) False]{
			\includegraphics[scale = 0.6]{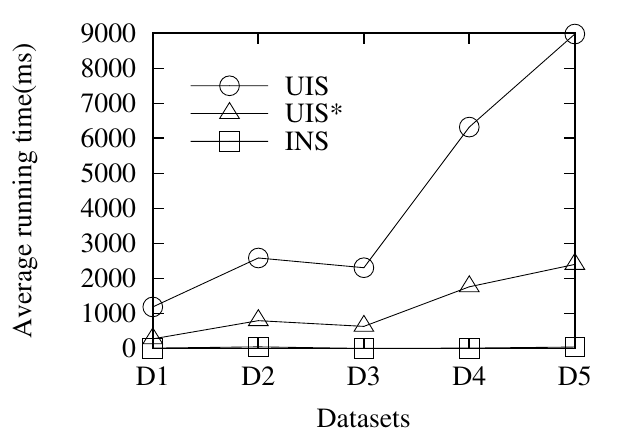}}
		\subfigure[(c) True]{
			\includegraphics[scale = 0.6]{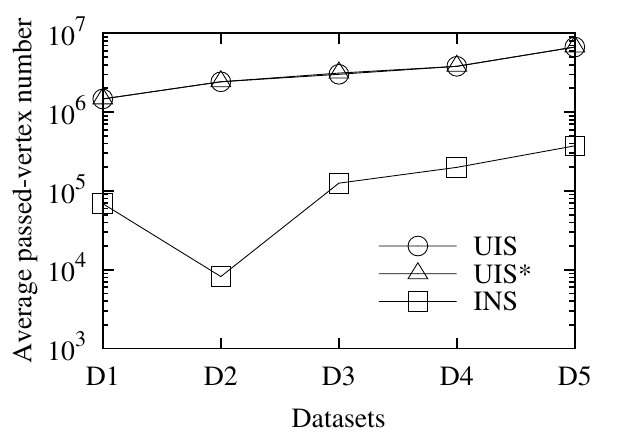}}
		\subfigure[(d) False]{
			\includegraphics[scale = 0.6]{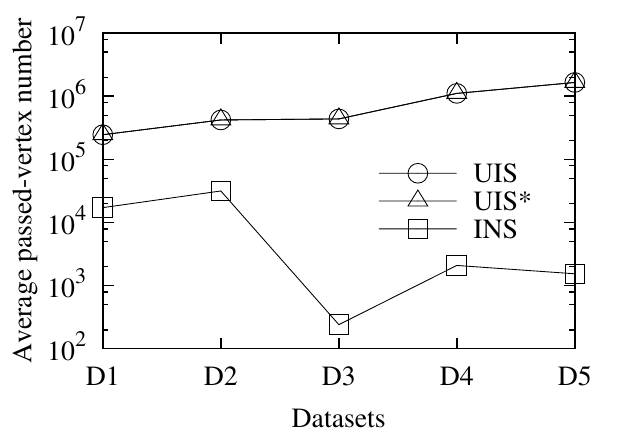}}
		\caption{Results related to the substructure constraint S5.}
		\label{fig:lubm_s5}
	\end{figure}
	
	\begin{figure}[t]
		\centering
		\renewcommand{\thesubfigure}{}
		\subfigure[(a) True]{
			\includegraphics[scale = 0.6]{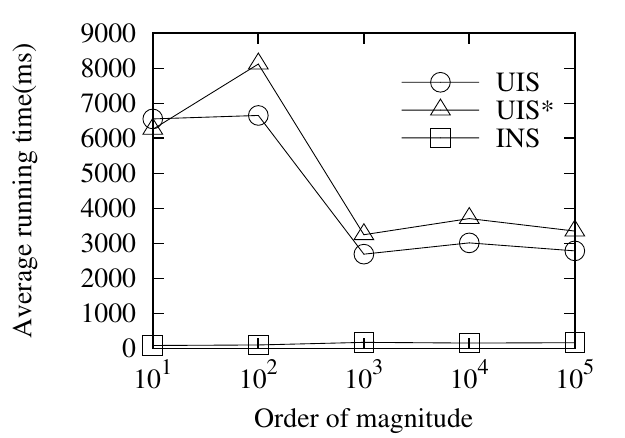}}
		\subfigure[(b) False]{
			\includegraphics[scale = 0.6]{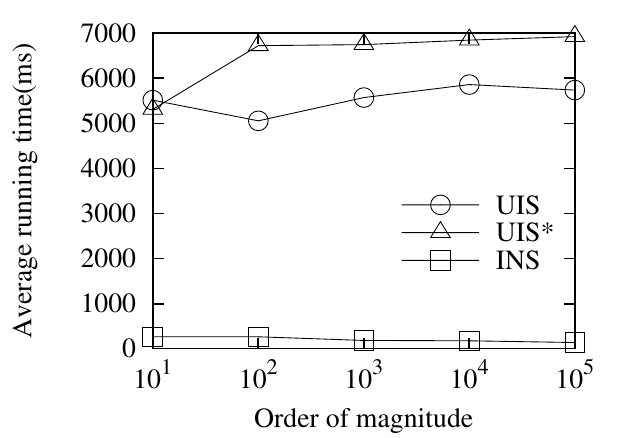}}
		\subfigure[(c) True]{
			\includegraphics[scale = 0.6]{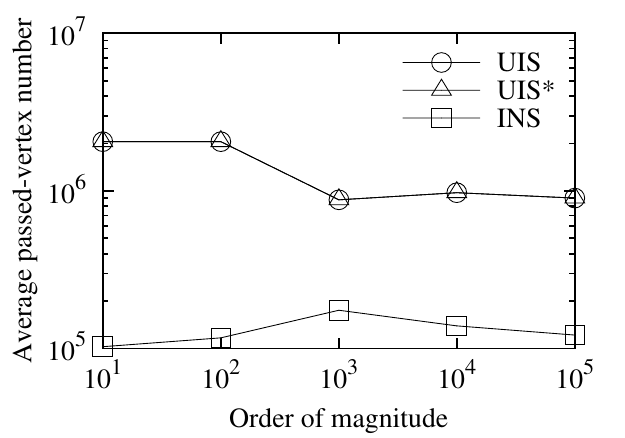}}
		\subfigure[(d) False]{
			\includegraphics[scale = 0.6]{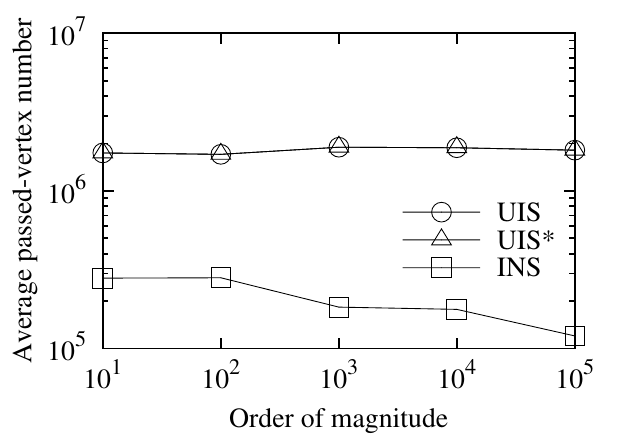}}
		\caption{Results related to random substructure constraints.}
		\label{fig:yago}
	\end{figure}

	The test query generation process of $\mathcal{Q}_t$ and $\mathcal{Q}_f$ is presented as follows. Firstly, assuming $Q\negthickspace=\negthickspace(s,t,L,S)$, both the way of obtaining $s$, $t$ and $L$, and the method of controlling the irrelevant variables are as illustrated in Section~\ref{section:r:lubm:query_generation}. Then, for the substructure constraint $S$ of $G$, we let $|V(S,G)|$ in the range of $[0.8m,1.2m]$, and randomly select an instance vertex $v$ in $G$~(Section~\ref{section:background}). According to the RDF schema of $G$ and the edges that are incident to $v$, we initialize $S\negthickspace=\negthickspace(?x, V_S, E_S, E_?)$ with a small selectivity, where $v\in V(S,G)$. After that, based on the gap between $|V(S,G)|$ and $[0.8m,1.2m]$, we gradually and randomly adjust $V_S$, $E_S$ and $E_?$. Finally, whether $Q$ belongs to $\mathcal{Q}_t$, or $\mathcal{Q}_f$, or we discard $Q$, is also as demonstrated in Section~\ref{section:r:lubm:query_generation}.
	
	\textbf{Experimental result analysis.} The results of the above experiments are depicted in Figure~\ref{fig:yago}. Assuming $\mathcal{Q}(r,m)$ represents a group of evaluation queries, where: (\romannumeral1) if $r=T$, each query in $\mathcal{Q}(r,m)$ is a query true, otherwise false; (\romannumeral2) $m$ denotes the order of magnitude as mentioned above. The average running time of UIS shows a downward trend from the experiments on $\mathcal{Q}(T,10^1)$~(about $7$s) to that on $\mathcal{Q}(T,10^5)$~(about $3$s), as drawn in Figure~\ref{fig:yago}(a). For the false queries~(Figure~\ref{fig:yago}(b)), the average running time of UIS does not change significantly from $\mathcal{Q}(F,10^1)$ to $\mathcal{Q}(F,10^5)$, which is approximately stable at $5.5$s. Additionally, its change trend of the average passed-vertex number is similar to that of the average running time, according to Figures~\ref{fig:yago}(c)-\ref{fig:yago}(d). 
	
	The query performance of UIS$^*$ is still worse than that of UIS, especially in the experiments related to the false queries~(Figure~\ref{fig:yago}(c)), while, the passed-number of UIS$^*$ is almost equal to that of UIS. Such results further prove that UIS$^*$ has a higher percentage of traversing the vertices in the input KG repeatedly. With the local index pruning the search space and accelerating query processing, the performance of INS is the best throughout the experiments. The average online search time consumption of INS is about one thousandth of either UIS or UIS$^*$ on both true-LSCR and false-LSCR queries. 
	
	However, the above results seems different from that in Section~\ref{section:r:lubm}, where we observe $|V(S,G)|$ highly dominates the query efficiency of a LSCR query $Q\negthickspace=\negthickspace(s,t,L,S)$. That is because the potential search spaces of the randomly generated evaluation queries are affected by $m$: (\romannumeral1) for a group of true queries $\mathcal{Q}(T,m)$, a larger $m$ means more vetices in the input KG that satisfy the randomly generated substructure constraint, which makes searching from a vertex to another relatively easier, so the efficiency shows a slight increase; (\romannumeral2) for a group of false queries $\mathcal{Q}(F,m)$, as the search algorithms are runs in the space that a vertex $s$ reaches to under the given substrucuture constraint, in which each vertex can be traversed at most twice, the efficiency can be hardly affected by $m$.
	
	\subsection{Discussion}\label{section:r:discussion}
	For the experimental results demonstrated in Section~\ref{section:r:lubm} and Section~\ref{section:r:yago}, we observe that: (\romannumeral1) comparing to \cite{Valstar17Landmark}, the indexing time and space consumption are affordable; (\romannumeral2) comparing to UIS and UIS$^*$, INS achieves a great performance on the different scales of the input KGs, even though, on a large KG, different degrees of the query performance degradation exist in the three algorithms; (\romannumeral3) the selectivity of a substructure constraint $S$ hardly affects the query performances of the proposed three algorithms, on a LSCR query $Q$$=$$(s,t,L,S)$, while, $|V(S,G)|$ is a main factor that dominates the query efficiency.

\section{Conclusion} \label{section:conclusion}
	This paper introduces a new variant of the reachability problem on KGs, i.e. LSCR queries, which  contains both label and substructure constraints. The LSCR queries on KGs is much more complicated than the existing reachability queries, which only consider the label constraints. On the one hand, the query processing time complexity of applying traditional online search strategies~(DFS/BFS) cannot be afforded. On the other hand, indexing a relatively large KG with the techniques of LCR queries is unrealistic, as the indexing time grows exponentially with the input KG scale. This work first presents a baseline method~(UIS) for LSCR queries not only on KGs but also on the general graphs. Then after introducing an intuition idea of implementing a SPARQL engine to address our problem, we devise an informed search algorithm INS with a \textit{local index}. Our experimental evaluation on both synthetic and real datasets confirms that the indexing time and space consumption of the local index is affordable, and, based on local index, INS solves the LSCR queries on KGs efficiently.

\ifCLASSOPTIONcompsoc
\section*{Acknowledgments}
This paper was partially supported by NSFC grant U1509216, U1866602, 61602129 and Microsoft Research Asia. 
\else
  \section*{Acknowledgment}
\fi

\ifCLASSOPTIONcaptionsoff
  \newpage
\fi

\bibliographystyle{plain}

%
%
\begin{IEEEbiography}[{\includegraphics[width=1in, height=1.2in,clip,
keepaspectratio]{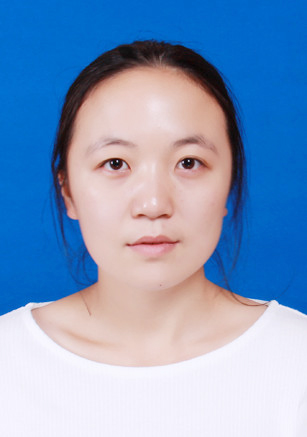}}]{Xiaolong Wan} received the bachelor's degree from the Harbin Institute of Technology, in 2018. She is working toward the PhD degree in the School of Computer Science and Technology, Harbin Institute of Technology. Her research interests include knowledge graph, graph mining.
\end{IEEEbiography}

\begin{IEEEbiography}[{\includegraphics[width=1in, height=1.2in,clip,
keepaspectratio]{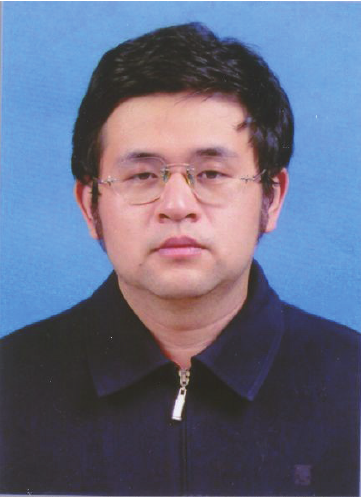}}]{Hongzhi Wang} is a professor and doctoral supervisor with the Harbin Institute of Technology. He was awarded a Microsoft fellowship and IBM PhD fellowship, and he was designated as a Chinese excellent database engineer. His research interests include big data management, data quality, graph data management, and web data management.
\end{IEEEbiography}

%




\end{document}